\documentclass[11pt,english]{article}
\usepackage{mathptmx}
\usepackage[T1]{fontenc}
\usepackage[latin9]{inputenc}
\usepackage{array}
\usepackage{float}
\usepackage{booktabs}
\usepackage{mathtools}
\usepackage{enumitem}
\usepackage{bm}
\usepackage{amsmath}
\usepackage{amsthm}
\usepackage{amssymb}
\usepackage{graphicx}
\usepackage[a4paper]{geometry}
\usepackage{setspace}
\usepackage[authoryear]{natbib}
\makeatletter

\providecommand{\tabularnewline}{\\}

\newlength{\lyxlabelwidth}      
	\newenvironment{elabeling}[2][]%
	{\settowidth{\lyxlabelwidth}{#2}
		\begin{description}[font=\normalfont,style=sameline,
			leftmargin=\lyxlabelwidth,#1]}
	{\end{description}}
\theoremstyle{plain}
\newtheorem{lem}{\protect\lemmaname}
\theoremstyle{definition}
\newtheorem{defn}{\protect\definitionname}
\theoremstyle{remark}
\newtheorem*{rem*}{\protect\remarkname}
\theoremstyle{definition}
 \newtheorem{example}{\protect\examplename}
\theoremstyle{plain}
\newtheorem{thm}{\protect\theoremname}
\theoremstyle{plain}
\newtheorem{cor}{\protect\corollaryname}

\usepackage{algorithm,algpseudocode}

\usepackage{caption}

\numberwithin{equation}{section}

\newcommand*\dif{\mathop{}\!\mathrm{d}}

\usepackage{enumitem}  
\setlist[itemize]{noitemsep,topsep=2pt}
\setlist{noitemsep,topsep=2pt} 

\usepackage[font=small,labelfont=bf]{caption}

\usepackage{enumitem}
\setlist[enumerate,1]{label=(\alph*),font=\normalfont}

\makeatother

\usepackage{babel}
\providecommand{\corollaryname}{Corollary}
\providecommand{\definitionname}{Definition}
\providecommand{\examplename}{Example}
\providecommand{\lemmaname}{Lemma}
\providecommand{\remarkname}{Remark}
\providecommand{\theoremname}{Theorem}

\begin{document}
\title{Beating the Best Constant Rebalancing Portfolio in Long-Term Investment:
A Generalization of the Kelly Criterion and Universal Learning Algorithm
for\\ Markets with Serial Dependence}
\author{Duy Khanh Lam\thanks{This work benefited from the dataset provided by my advisor, Giulio
Bottazzi at Scuola Superiore Sant'Anna, which was used in the experiments.}\\Scuola Normale Superiore\\ \\(Working paper, 1st draft)}
\maketitle
\begin{abstract}
In the online portfolio optimization framework, existing learning
algorithms generate strategies that yield significantly poorer cumulative
wealth compared to the best constant rebalancing portfolio in hindsight,
despite being consistent in asymptotic growth rate. While this unappealing
performance can be improved by incorporating more side information,
it raises difficulties in feature selection and high-dimensional settings.
Instead, the inherent serial dependence of assets' returns, such as
day-of-the-week and other calendar effects, can be leveraged. Although
latent serial dependence patterns are commonly detected using large
training datasets, this paper proposes an algorithm that learns such
dependence using only gradually revealed data, without any assumption
on their distribution, to form a strategy that eventually exceeds
the cumulative wealth of the best constant rebalancing portfolio.\smallskip

Moreover, the classical Kelly criterion, which requires independent
assets' returns, is generalized to accommodate serial dependence in
a market modeled as an independent and identically distributed process
of random matrices. In such a stochastic market, where existing learning
algorithms designed for stationary processes fail to apply, the proposed
learning algorithm still generates a strategy that asymptotically
grows to the highest rate among all strategies, matching that of the
optimal strategy constructed under the generalized Kelly criterion.
The experimental results with real market data demonstrate the theoretical
guarantees of the algorithm and its performance as expected, as long
as serial dependence is significant, regardless of the validity of
the generalized Kelly criterion in the experimental market. This further
affirms the broad applicability of the algorithm in general contexts.\bigskip
\end{abstract}
{\small\textit{Keywords}}{\small : Online Learning, Universality, Dynamic
Strategy, Generalized Kelly Criterion, Optimal Growth, Latent Patterns.}{\small\par}

\pagebreak\setlength{\abovedisplayskip}{2.5pt} 
\setlength{\abovedisplayshortskip}{2.5pt}
\setlength{\belowdisplayskip}{2.5pt} 
\setlength{\belowdisplayshortskip}{2.5pt} 

\section{Introduction}

In the periodically repeated game of the online portfolio optimization
framework, the goal of the investor is to construct a strategy as
a sequence of different portfolios, based only on past assets' returns,
such that it can grow as fast as the best constant strategy in hindsight,
which is formed by a single invariant portfolio, as more time periods
elapse. Although many learning algorithms surveyed in \citet{Erven2020}
have been proposed for strategy construction, the cumulative wealth
they yield appears to be significantly lower than that of the best
retrospective constant strategy, even though they still grow asymptotically
at the same rate. This makes them less appealing in practice, as their
cumulative wealth may eventually perform worse than the naive strategy
of buying and holding a single asset over time. A natural solution
that arises is to enhance the strategy by learning with more information
rather than relying solely on past assets' returns, such as the algorithms
that learn from finite-state side information online, as proposed
by \citet{Cover1996a}, and extended to continuous side information
by \citet{Bhatt2023}.\smallskip

However, leveraging side information inherently faces difficulties
in selecting external features that are truly useful among a large
pool of choices and in determining a proper size of them for the learning
process, not to mention the challenges in dealing with continuous
information. Instead, a type of implicit information that is inherently
available can be utilized for strategy construction, namely, the serial
dependence, i.e., the pattern of the assets' returns. In order to
detect a latent pattern, a large amount of available data is often
required for usable information extraction; thus, it seems challenging
to learn such patterns online, as the data is revealed gradually rather
than being available as a whole. Therefore, in such a market with
serial dependence, which deviates from being independent and identically
distributed (i.i.d.) or stationary, two interesting problems are brought
to attention immediately and partly constitute the motivations of
this paper: what is the form of the strategy that yields the highest
asymptotic growth rate, and how can the implicit information be learned
from the gradually revealed data to construct such an optimal strategy.\smallskip

\textbf{Overview of main results}. This paper discusses and proposes
a solution to the motivated challenges as mentioned. An online learning
algorithm is proposed, based only on past assets' returns, to construct
a dynamic strategy such that its yielded cumulative wealth can exceed
that yielded by the best constant strategy in hindsight over time,
without any statistical assumption on the market. This is achieved
through utilizing the serial dependence of the market by taking advantage
of the discrepancy in the characteristics of the subsequences of assets'
returns, in order to guarantee that the constructed strategy can grow
asymptotically faster than the best constant strategy in hindsight.
Moreover, in a stochastic market exhibiting serial dependence as a
block-wise i.i.d. process, a generalization of the classical Kelly
criterion, which was originally proposed by \citet{Kelly1956} and
is well known to be valid only in the i.i.d. market, is established
to determine a dynamic strategy that yields the highest asymptotic
growth rate among all strategies. In such a stochastic market but
with an unknown joint distribution, while the optimal strategy according
to the generalized Kelly criterion is not identifiable and the online
learning algorithms requiring stationarity of the process are not
applicable, the strategy constructed according to the proposed learning
algorithm is still able to achieve the highest asymptotic growth rate.\smallskip

\textbf{Paper organization and novelties}. To emphasize the generality
and broad applicability of the proposed learning algorithm, the paper's
sections are organized in an order where the algorithm's utilization
of the serial dependence of the assets' returns serves as the primary
motivation, rather than the asymptotic optimality of the strategy
and the generalization of the Kelly criterion. Specifically, they
are outlined along with their respective novelties as follows:\smallskip
\begin{elabeling}{00.00.0000}
\item [{\textit{Section~2}:}] Based on the brief summary recalling the
Universal Portfolio algorithm and the benchmarking constant strategy
proposed by \citet{Cover1991} in Section 2.1 as foundational concepts,
a novel online learning algorithm for strategy construction and its
associated $k$-\emph{cyclic constant strategy} are proposed in Section
2.2 as a significant extension to leverage the implicit serial dependence
of assets' returns. Then, the original sublinear upper bound established
for the Universal Portfolio strategy and the best constant strategy
in hindsight is extended to the proposed strategies, respectively,
which guarantees their consistency in growth rate over time. This
enables the strategy constructed according to the proposed algorithm
to exceed the best constant strategy in cumulative wealth over time,
which is essentially impossible for the Universal Portfolio strategy.$\vspace{0.5ex}$
\item [{\textit{Section~3}:}] In order to extend the classical Kelly criterion,
the concept of $k$-\emph{log-optimality}, defined as the condition
for constructing an optimal $k$-cyclic constant strategy with respect
to the joint distribution of the assets' returns, is formalized as
a generalization of the classical criterion to accommodate serial
dependence in a block-wise i.i.d. market represented by a stochastic
process of return matrices. Accordingly, the asymptotic optimality
in growth rate originally established for the classical criterion
by \citet{Breiman1960,Breiman1961} is extended to the generalized
one. As a consequence, in this market, without requiring knowledge
of the joint distribution, the proposed algorithm in Section 2 generates
a strategy that grows to the highest asymptotic rate among all dynamic
strategies, matching that of the optimal $k$-cyclic constant strategy.$\vspace{0.5ex}$
\item [{\textit{Section~}4:}] The proposed learning algorithm is tested
on a real market dataset of $6798$ days, with tuning of the parameter
$k$ to investigate its impact on the constructed strategy\textquoteright s
performance. The results demonstrate that the tested strategy outperforms
the Universal Portfolio strategy in terms of cumulative wealth for
all choices of $k$. With proper tuning of $k$, the cumulative wealth
yielded by the constructed strategy exceeds that of the best constant
strategy in hindsight. Moreover, the results show that the asymptotic
growth rate of the Universal Portfolio strategy is not the highest,
thereby invalidating the Kelly criterion in the experimental market.
This observation supports the existence of significant serial dependence
in the assets' returns process, which is further strengthened by evidence
that the best $k$-cyclic constant strategy yields higher cumulative
wealth than that yielded by the best constant strategy, with a large
gap as $k$ increases.
\end{elabeling}

\section{Model setting, key concepts and strategy proposal}

Consider a number $m\geq2$ of risky assets in the stock market during
discrete time periods $n\in\mathbb{N}_{+}$. For convenience in expressing
assets' returns later, let the vector $p_{n}=\big(p_{n,1},...,p_{n,m}\big)$
denote the prices of $m$ assets at time $n$, where $p_{n,j}\in\mathbb{R}_{++}$
is the price of the $j$-th asset. Subsequently, the real-valued random
vector $X_{n}\coloneqq\big(X_{n,1},...,X_{n,m}\big)\in\mathbb{R}_{++}^{m}$
represents the assets' returns at time $n$, while the vector $x_{n}\coloneqq\big(x_{n,1},...,x_{n,m}\big)$
denotes its corresponding realization, where $x_{n,j}\coloneqq p_{n,j}/p_{n-1,j}$
for the $j$-th asset. Accordingly, let $x_{1}^{n}\coloneqq\big\{ x_{i}\big\}_{i=1}^{n}$
be shorthand for the sequence of realizations of the respective sequence
$\big\{ X_{i}\big\}_{i=1}^{n}$, so that an infinite sequence of realizations
is denoted as $x_{1}^{\infty}$. The simplex $\mathcal{B}^{m}\coloneqq\big\{\beta\in\mathbb{R}^{m}:\,{\displaystyle {\textstyle \sum}}_{{\scriptstyle j=1}}^{{\scriptstyle m}}\beta_{j}=1,\,\beta_{j}\geq0\big\}$
represents the domain of all selectable no-short portfolios, and the
notation $\big<b,x_{n}\big>$, where $\big<\cdot,\cdot\big>$ denotes
the scalar product of two vectors, represents the return of a portfolio
$b\in\mathcal{B}^{m}$ with respect to $x_{n}$.\smallskip

Since the future assets' returns are unforeseeable, a portfolio selection
at any time is made only causally depending on the observed data;
thus, let $b_{n}:\mathbb{R}_{++}^{m\times(n-1)}\to\mathcal{B}^{m}$
denote the portfolio selected at time $n$ as a function of the past
realizations $x_{1}^{n-1}$. An infinite sequence of portfolios made
over time, denoted as $\big(b_{n}\big)\coloneqq\big\{ b_{n}\big\}_{n=1}^{\infty}$,
is termed a \textit{strategy}; thus, if the strategy takes the form
of a fixed portfolio $b$ over time, it is called a \emph{constant
strategy} and denoted as $\big(b\big)$, without time index, to distinguish
it from a dynamic one. Accordingly, let the initial fortune be $S_{0}=1$
by convention, and assuming no commission fees, the \emph{cumulative
wealth} and the \emph{growth rate} after $n$ time periods of investment
yielded by a generic strategy $\big(b_{n}\big)$ are defined respectively
as:
\[
S_{n}\big(b_{n}\big)\coloneqq{\displaystyle {\displaystyle \prod_{i=1}^{n}}\big<b_{i},x_{i}\big>}\text{ and }W_{n}\big(b_{n}\big)\coloneqq\dfrac{1}{n}\log S_{n}\big(b_{n}\big)={\displaystyle \dfrac{1}{n}\sum_{i=1}^{n}\log\big<b_{i},x_{i}\big>.}
\]
Here, $S_{n}\big(b_{n}\big)$ and $W_{n}\big(b_{n}\big)$ serve as
shorthands, implicitly depending on the realization sequence $x_{1}^{n}$
of the assets' returns, without making this explicit in the notation.
Noting that, since the assets' returns are positive and the strategy
is self-financed, the resulting $S_{n}\big(b_{n}\big)$ and $W_{n}\big(b_{n}\big)$
are always positive, thus the investment game lasts indefinitely.

\subsection{Brief on the best constant strategy, Universal Portfolio, and consistency}

Before introducing the extended benchmarking strategy and proposing
a learning algorithm in the next section, this section is dedicated
to recalling the essential foundation and basic concepts of the Universal
Portfolio and its corresponding benchmarking strategy. As first posed
by \citet{Cover1991}, for each sequence of realizations of assets'
returns $x_{1}^{n}$, we can identify an optimal constant strategy,
i.e., a fixed portfolio, that solves $\max_{b\in\mathcal{B}^{m}}S_{n}\big(b\big)$,
which is considered the benchmarking strategy, so that a constructed
dynamic strategy $\big(b_{n}\big)$ should asymptotically approach
it in growth rate over time. This goal is formally termed the \emph{consistency}
criterion, as follows:
\begin{equation}
\lim_{n\to\infty}\min_{b\in\mathcal{B}^{m}}\big(W_{n}\big(b_{n}\big)-W_{n}\big(b\big)\big)=0.\label{consistency}
\end{equation}
To avoid misunderstanding that the realizations of the random vectors
of assets' returns imply they are known, it is worth noting that the
infinite sequence $x_{1}^{\infty}$ is treated as deterministic but
unknown.\smallskip

Since the best constant strategy at any time $n$ is only known in
hindsight, the strategy $\big(b_{n}\big)$ must guarantee the consistency
condition (\ref{consistency}) universally for any infinite sequence
of unforeseeable realizations $x_{1}^{\infty}$. The following Universal
Portfolio algorithm constructs such a desired strategy that employs
only the available observed realizations of assets' returns:
\begin{equation}
b_{n}\coloneqq{\displaystyle \dfrac{{\displaystyle \int_{\mathcal{B}^{m}}{\displaystyle b}}S_{n-1}\big(b\big)\mu\big(b\big)\dif b}{{\displaystyle \int_{\mathcal{B}^{m}}S_{n-1}\big(b\big)\mu\big(b\big)\dif b}}},\,\forall n,\label{UP algorithm}
\end{equation}
where $\mu(\cdot)$ denotes a probability density over the simplex
$\mathcal{B}^{m}$. Noting that if the density $\mu(\cdot)$ is symmetric,
then we can write in particular $b_{1}={\displaystyle \int_{\mathcal{B}^{m}}{\displaystyle b}}\mu\big(b\big)\dif b=\big(1/m,...,1/m\big)$.\smallskip

By construction, the Universal Portfolio strategy $\big(b_{n}\big)$,
as described above, regardless of the choice of the density $\mu(\cdot)$,
yields a cumulative wealth that is always strictly bounded by that
of the best constant strategy in hindsight at any time $n$, as follows:
\begin{align}
S_{n}\big(b_{n}\big)={\displaystyle \prod_{i=1}^{n}\dfrac{{\displaystyle \int_{\mathcal{\mathcal{B}}^{m}}\big<b,x_{i}\big>}S_{i-1}\big({\displaystyle b\big)}\mu\big(b\big)\dif b}{{\displaystyle \int_{\mathcal{\mathcal{B}}^{m}}S_{i-1}\big({\displaystyle b\big)}\mu\big(b\big)\dif b}}} & {\displaystyle =\int_{\mathcal{\mathcal{B}}^{m}}S_{n}\big({\displaystyle b\big)}\mu\big(b\big)\dif b\text{ \,\,\,\,\,\,\,\,\,(by telescoping)}}\nonumber \\
 & <\max_{b\in\mathcal{B}^{m}}S_{n}\big(b\big)\int_{\mathcal{\mathcal{B}}^{m}}\mu\big(b\big)\dif b=\max_{b\in\mathcal{B}^{m}}S_{n}\big(b\big),\,\forall n.\label{UP bound}
\end{align}
This is the disadvantage mentioned in the introduction, namely that
the cumulative wealth of the Universal Portfolio strategy is not as
good as expected in practice compared with its benchmarking strategy.
Nonetheless, the strategy $\big(b_{n}\big)$ can satisfy the consistency
criterion by an appropriate choice of density $\mu(\cdot)$. This
is guaranteed by Lemma \ref{Lemma 1}, which shows a sublinear upper
bound in the worst case of the realization sequence $x_{1}^{n}$ for
the difference $\max_{b\in\mathcal{B}^{m}}\left|\log S_{n}\big(b_{n}\big)-\log S_{n}\big(b\big)\right|$.
\begin{lem}
Consider the Universal Portfolio strategy $\big(b_{n}\big)$ constructed
according to (\ref{UP algorithm}). If $\mu(\cdot)$ is chosen as
the uniform density, then:\label{Lemma 1}
\[
{\displaystyle \max_{b\in\mathcal{B}^{m}}\max_{x_{1}^{n}}\big(\log S_{n}\big(b\big)-\log S_{n}\big(b_{n}\big)\big)\leq\big(m-1\big)\log\big(n+1\big),\text{ \ensuremath{\forall n}}}.
\]
Alternatively, if $\mu(\cdot)$ is chosen as the Dirichlet $\big(1/2,...,1/2\big)$
density, then:
\[
{\displaystyle \max_{b\in\mathcal{B}^{m}}\max_{x_{1}^{n}}\big(\log S_{n}\big(b\big)-\log S_{n}\big(b_{n}\big)\big)\leq\frac{m-1}{2}\log\big(n+1\big)+\log2,\text{ \ensuremath{\forall n}}}.
\]
\end{lem}
\begin{proof}
See the proofs in \citet{Cover1991,Cover1996a}.
\end{proof}

\subsection{The best $\bm{k}$-cyclic constant strategy and universal learning
algorithm}

The mechanism of the best constant strategy is to utilize solely the
variation of the assets' returns over time periods, which can be generalized
to address the inherent drawback of the Universal Portfolio strategy
described in (\ref{UP bound}), without requiring additional side
information. Specifically, for any sequence of assets' returns, we
can partition it into multiple subsequences with distinct characteristics.
This partition method allows for better exploitation of the inherent
discrepancy between each time within a fixed window horizon of investment.
The formal definition of the $k$-cyclic constant strategy, for any
integer $k\geq1$, interpreted as the cycle length of a segment of
time, is introduced in Definition \ref{k-wise strategy definition}.
Then, Example \ref{Example k wise strategy} provides some instances
of $k$-cyclic strategies in real-world practice contexts to clarify
its definition and illustrate its conceptual motivation, while its
growth rate behavior in relation to the choice of length $k$ is considered
later in Theorem \ref{Theorem 1}.
\begin{defn}
Given an integer $k\geq1$ as the cycle length, a $k$-cyclic constant
strategy $\big(b_{n}^{k}\big)$ is defined by a looping collection
of fixed portfolios $\big(b^{1},...,b^{k}\big)$, where $b^{i}\in\mathcal{B}^{m}$
for all $i\in\big\{1,...,k\big\}$, as:\label{k-wise strategy definition}
\[
\big(b_{n}^{k}\big)\coloneqq\big\{\big(b_{kt+1},b_{kt+2},...,b_{kt+k}\big)\big\}_{t=0}^{\infty},\text{ where }b_{kt+i}\coloneqq b^{i},\,\forall i\leq k,\,\forall t\geq0.
\]
In the special case of the 1-cyclic constant strategy, it is equivalent
to a constant strategy.\medskip

As the main motivation behind the concept in Definition \ref{k-wise strategy definition},
the number $k$ of a $k$-cyclic constant strategy $\big(b_{n}^{k}\big)$
is a changeable parameter that serves to tune the corresponding strategy
with a greater degree of freedom compared to the basic constant strategy.
In detail, since the considered assets yield returns that fluctuate
over time, taking values both above and below one, a constant strategy
that allocates capital according to a fixed portfolio regardless of
return levels inherently limits its ability to adjust when returns
decline. In contrast, the $k$-cyclic constant strategy relaxes this
constraint by allowing rebalancing across several fixed portfolios
over time. If $k$ is set equal to $n$ at the observed time $n$,
the $k$-cyclic constant strategy allows investing the entire periodic
capital into the assets with the highest returns from time $1$ to
time $n$; however, this practice is infeasible, as the future market
is unforeseeable. Instead, $k$ must be chosen sufficiently small
relative to the investment period length, as in the practices illustrated
in Example \ref{Example k wise strategy}, so that the number of generated
cycles of length $k$ is sufficiently large.
\end{defn}
\begin{rem*}
In practice, a $k$-cyclic constant strategy is flexible in its cycle
indexation, as the time periods $n$ need not be equally spaced. Rather
than indexing each trading day as a separate period, as mentioned
in Example \ref{Example k wise strategy}, an investor can group three
consecutive days, from Tuesday to Thursday, as a single period to
capture anomalies that regularly occur on one of these days.\smallskip
\end{rem*}
\begin{example}
An investor applies a $5$-cyclic constant strategy defined by a fixed
weekly routine of portfolio allocations, designed to exploit systematic
volatility and anomalies across trading days, which is commonly referred
to as the day-of-the-week effect in the literature. This practice
reflects the investor\textquoteright s belief that a hidden pattern
exists in the daily assets' returns within a week, which can be exploited
in constructing a portfolio. Typical observations include the tendency
of investors to sell on Fridays and buy on Mondays, or that returns
tend to be stronger in the middle of the week. Specifically, a distinct
portfolio is assigned to each day from Monday to Friday, and this
routine is repeated consistently throughout the investment period.
As another example in a relatively higher-frequency context, it is
also possible to exploit volatility over much shorter intervals, such
as hours. In this case, the investor may assign a fixed portfolio
to a specific hour within a trading day and repeat this allocation
daily.\label{Example k wise strategy}
\end{example}
By definition, a $k$-cyclic constant strategy $\big(b_{n}^{k}\big)$
for any $k\geq2$ is a dynamic strategy, rather than a constant one
defined simply by a fixed portfolio. Therefore, the best $k$-cyclic
constant strategy at any time period $n$ is determined by a collection
of fixed portfolios that solves the following cumulative wealth maximization
problem:
\[
\max_{(b^{1},...,b^{k})\in\mathcal{B}^{m\times k}}S_{n}\big(b_{n}^{k}\big)=\max_{(b^{1},...,b^{k})\in\mathcal{B}^{m\times k}}{\displaystyle \prod_{t\geq0:\,kt+i\leq n}}\big<b^{i},x_{kt+i}\big>,\,\forall i\leq\min\big\{ k,n\big\},\,\forall n,
\]
which, similarly to the best constant strategy, is only known in hindsight
as it depends on the sequence of realizations of assets' returns $x_{1}^{n}$.
An immediate consequence is that the best $k$-cyclic constant strategy
yields strictly greater cumulative wealth and growth rate than the
best constant strategy at any time $n$, for all $k\geq2$. Accordingly,
since the best $k$-cyclic constant strategy serves as the benchmark,
the original consistency criterion (\ref{consistency}) for a potential
strategy $\big(b_{n}\big)$, defined in the case $k=1$, is generalized
to the following for:
\begin{equation}
\lim_{n\to\infty}\min_{(b^{1},...,b^{k})\in\mathcal{B}^{m\times k}}\big(W_{n}\big(b_{n}\big)-W_{n}\big(b_{n}^{k}\big)\big)=0.\label{k-cyclic consistency}
\end{equation}
In the context of unforeseeable future realizations of assets' returns,
the subsequent algorithm is proposed to construct a strategy that
guarantees the desired consistency (\ref{k-cyclic consistency}),
using only the observed assets' returns over time.\smallskip

\textbf{Strategy proposal}. Consider a strategy $\big(b_{n}\big)$
constructed as follows:
\begin{equation}
b_{i}\coloneqq\Big(\frac{1}{m},...,\frac{1}{m}\Big)\text{ and }{\displaystyle b_{kt+i}\coloneqq{\displaystyle \dfrac{{\displaystyle \int_{\mathcal{B}^{m}}{\displaystyle b}}{\displaystyle {\displaystyle \prod_{j=0}^{t-1}}\big<b,x_{kj+i}\big>}\mu\big(b\big)\dif b}{{\displaystyle \int_{\mathcal{B}^{m}}{\displaystyle \prod_{j=0}^{t-1}}\big<b,x_{kj+i}\big>\mu\big(b\big)\dif b}}},\,\forall i\leq k,\,\forall t\geq1,}\label{proposal strategy}
\end{equation}
where $\mu(\cdot)$ is chosen either as the uniform density or the
Dirichlet $\big(1/2,...,1/2\big)$ density.\smallskip

A closer inspection of the above algorithm\textquoteright s mechanism
reveals that it transforms the original repeated game, in which the
best constant strategy serves as the benchmark, into a simultaneously
repeated game. Specifically, since a $k$-cyclic constant strategy
simultaneously loops over a set of $k$ portfolios, the algorithm
decomposes the single investment game into $k$ parallel repeated
games, each of which identifies one of the $k$ best constant strategies
at any time $n\geq k$. Based on this decomposition, the algorithm
(\ref{proposal strategy}) employs the Universal Portfolio strategy
defined in (\ref{UP algorithm}) as a recursive subroutine, applied
independently to each of the $k$ subsequences of the observed realizations
$x_{1}^{n}$ of assets' returns. Consequently, the strategy $\big(b_{n}\big)$
constructed by this algorithm satisfies the consistency criterion
(\ref{k-cyclic consistency}), analogous to how the Universal Portfolio
strategy satisfies the original consistency condition (\ref{consistency}).
This is demonstrated in Theorem \ref{Theorem 1}, which establishes
a sublinear upper bound in the worst case of the realization sequence
$x_{1}^{n}$ for the difference $\max_{(b^{1},...,b^{k})\in\mathcal{B}^{m\times k}}\left|\log S_{n}\big(b_{n}\big)-\log S_{n}\big(b_{n}^{k}\big)\right|$.\footnote{By decomposing the sequence of realizations $x_{1}^{n}$ into $k$
subsequences, a desired strategy $\big(b_{n}\big)$ can be constructed
by applying learning algorithms surveyed in the paper \citet{Erven2020}.
Subsequently, an upper bound for the difference $\max_{(b^{1},...,b^{k})\in\mathcal{B}^{m\times k}}\left|\log S_{n}\big(b_{n}\big)-\log S_{n}\big(b_{n}^{k}\big)\right|$
can be established using a similar argument to the proof of Theorem
\ref{Theorem 1}, though not in the worst-case setting of $x_{1}^{n}$.
However, aside from differences in convergence speed, employing learning
algorithms other than the one proposed in (\ref{proposal strategy})
requires caution, as they often rely on assumptions applied over the
original sequence of realizations, which may not hold for each of
the $k$ decomposed subsequences.}
\begin{thm}
Consider the strategy $\big(b_{n}\big)$ constructed according to
(\ref{proposal strategy}), with any $k\geq2$. If $\mu(\cdot)$ is
chosen as the uniform density, then:\label{Theorem 1}
\[
0<{\displaystyle \max_{(b^{1},...,b^{k})\in\mathcal{B}^{m\times k}}\max_{x_{1}^{n}}\big(\log S_{n}\big(b_{n}^{k}\big)-\log S_{n}\big(b_{n}\big)\big)\leq k\big(m-1\big)\log\big(n+1\big),\text{ \ensuremath{\forall n}}}.
\]
Alternatively, if $\mu(\cdot)$ is chosen as the Dirichlet $\big(1/2,...,1/2\big)$
density, then:
\[
0<{\displaystyle \max_{(b^{1},...,b^{k})\in\mathcal{B}^{m\times k}}\max_{x_{1}^{n}}\big(\log S_{n}\big(b_{n}^{k}\big)-\log S_{n}\big(b_{n}\big)\big)\leq\frac{k}{2}\big(m-1\big)\log\big(n+1\big)+k\log2,\text{ \ensuremath{\forall n}}}.
\]
\end{thm}
\begin{proof}
To establish the required upper bound, it suffices to show that the
best $k$-cyclic constant strategy at any time $n\geq k$ can be identified
by detecting, separately, $k$ best constant strategies with respect
to $k$ subsequences. Indeed, since the following product factorization
holds:
\[
\max_{(b^{1},...,b^{k})\in\mathcal{B}^{m\times k}}S_{n}\big(b_{n}^{k}\big)=\prod_{i=1}^{\min\{k,n\}}\max_{b\in\mathcal{B}^{m}}{\displaystyle \prod_{t\geq0:\,kt+i\leq n}}\big<b,x_{kt+i}\big>,\,\forall n,
\]
the algorithm constructs a Universal Portfolio strategy for each infinite
subsequence $\big\{ x_{kt+i}\big\}_{t=0}^{\infty}$, where $1\leq i\leq k$,
of the infinite sequence $x_{1}^{\infty}$.\smallskip

Hence, if $\mu(\cdot)$ is chosen as the uniform density, applying
the upper bound in Lemma \ref{Lemma 1} yields the following bound
on the difference of interest:
\begin{align*}
0 & \leq\max_{(b^{1},...,b^{k})\in\mathcal{B}^{m\times k}}\max_{x_{1}^{n}}\big(\log S_{n}\big(b_{n}^{k}\big)-\log S_{n}\big(b_{n}\big)\big)\\
 & =\sum_{i=1}^{\min\{k,n\}}\max_{b\in\mathcal{B}^{m}}\max_{x_{1}^{n}}{\displaystyle \sum_{t\geq0:\,kt+i\leq n}\big(\log}\big<b,x_{kt+i}\big>-{\displaystyle \log}\big<b_{kt+i},x_{kt+i}\big>\big)\\
 & \leq\sum_{i=1}^{\min\{k,n\}}\big(m-1\big)\log\Big(\left\lfloor \frac{n-i}{k}+1\right\rfloor +1\Big)\\
 & \leq k\big(m-1\big)\log\big(n+1\big),\,\forall n,
\end{align*}
where the value $\left\lfloor \big(n-i\big)/k+1\right\rfloor $ represents
the number of terms in the set $\big\{ x_{kt+i}:\,kt+i\leq n\big\}$,
i.e., its cardinality. Meanwhile, the desired lower bound follows
directly from the property (\ref{UP bound}) of the Universal Portfolio
strategy. The same reasoning applies when $\mu(\cdot)$ is chosen
as the Dirichlet $\big(1/2,...,1/2\big)$ density, which completes
the proof.
\end{proof}
Since the best $k$-cyclic constant strategy with $k\geq2$ always
outperforms the best constant strategy in terms of cumulative wealth,
Theorem \ref{Theorem 1} implies that the strategy $\big(b_{n}\big)$,
constructed according to (\ref{proposal strategy}), can eventually
exceed the cumulative wealth of the best constant strategy in hindsight.
This is due to the following inequality in asymptotic growth rate:
\[
\limsup_{n\to\infty}\max_{b\in\mathcal{B}^{m}}\big(W_{n}\big(b\big)-W_{n}\big(b_{n}\big)\big)=\limsup_{n\to\infty}\max_{b\in\mathcal{B}^{m}}\min_{(b^{1},...,b^{k})\in\mathcal{B}^{m\times k}}\big(W_{n}\big(b\big)-W_{n}\big(b_{n}^{k}\big)\big)\leq0.
\]
Moreover, Corollary \ref{corollary 1} below highlights the motivation
for increasing the cycle length: the best $k$-cyclic constant strategy
guarantees strictly greater cumulative wealth than the best $h$-cyclic
constant strategy in hindsight at any time $n$, provided that $k$
is divisible by $h$. Consequently, the corresponding constructed
strategy $\big(b_{n}\big)$ associated with cycle length $k$ can
yield higher asymptotic performance than one constructed with cycle
length $h$. However, as the cycle length increases, the convergence
of the asymptotic growth rate of the strategy $\big(b_{n}\big)$ to
that of the best $k$-cyclic constant strategy may become slower due
to the larger upper bound established in Theorem \ref{Theorem 1}.
\begin{cor}
Consider two $k$-cyclic and $h$-cyclic constant strategies with
$k\equiv0\mod h$. Then:\label{corollary 1}
\[
\max_{(b^{1},...,b^{h})\in\mathcal{B}^{m\times h}}S_{n}\big(b_{n}^{h}\big)\leq\max_{(b^{1},...,b^{k})\in\mathcal{B}^{m\times k}}S_{n}\big(b_{n}^{k}\big),\,\forall n.
\]
Hence, the strategy $\big(b_{n}\big)$, constructed according to (\ref{proposal strategy})
and associated with $k$, satisfies:
\[
\limsup_{n\to\infty}\max_{(b^{1},...,b^{h})\in\mathcal{B}^{m\times h}}\big(W_{n}\big(b_{n}^{h}\big)-W_{n}\big(b_{n}\big)\big)\leq0.
\]
\end{cor}
\begin{proof}
Since $k\equiv0\mod h$, let $d=k/h$ denote the integer quotient.
Then,
\begin{align*}
S_{n}\big(b_{n}^{h}\big)=\prod_{i=1}^{h}{\displaystyle \prod_{t\geq0:\,ht+i\leq n}}\big<b^{i},x_{ht+i}\big>= & \prod_{i=1}^{h}{\displaystyle \prod_{j=0}^{d-1}}{\displaystyle \Big(\prod_{t\geq0:\,kt+i+jh\leq n}}\big<b^{i},x_{kt+i+jh}\big>\Big)\\
\leq & \prod_{i=1}^{h}{\displaystyle \prod_{j=0}^{d-1}}\max_{b\in\mathcal{B}^{m}}{\displaystyle \Big(\prod_{t\geq0:\,kt+i+jh\leq n}}\big<b,x_{kt+i+jh}\big>\Big)\\
= & \prod_{i=1}^{k}\max_{b\in\mathcal{B}^{m}}{\displaystyle \prod_{t\geq0:\,kt+i\leq n}}\big<b,x_{kt+i}\big>\\
= & \max_{(b^{1},...,b^{k})\in\mathcal{B}^{m\times k}}S_{n}\big(b_{n}^{k}\big).
\end{align*}
Thus, the remaining inequality immediately follows using Theorem \ref{Theorem 1},
completing the proof.
\end{proof}

\section{Generalized Kelly criterion for markets with serial dependence}

This section focuses on establishing the possible asymptotic optimality
in growth rate for the $k$-cyclic constant strategy, which will be
shown as a generalization of the classical Kelly criterion. Recall
that the Kelly criterion is widely known as the foundation for deriving
an optimal strategy in the i.i.d. market, specifically a constant
strategy that achieves the highest asymptotic growth rate among all
possible dynamic strategies. Such asymptotic optimality was originally
demonstrated in \citet{Breiman1960,Breiman1961} as an adaptation
to the stock market of the betting framework introduced in the classical
paper by \citet{Kelly1956}. Moreover, due to Corollary \ref{corollary 1},
there always exists an optimal $k$-cyclic constant strategy, for
any $k\geq2$, that attains the highest asymptotic growth rate, similar
to the optimal constant strategy under the Kelly criterion. This suggests
that the i.i.d. market is merely a special case of a more general
class of stochastic markets in which the highest asymptotic growth
rate among all dynamic strategies can be achieved by the optimal $k$-cyclic
constant strategy. Since this hypothetical stochastic market is broader
than the i.i.d. case, it is expected to allow for serial dependence
in the assets' returns.\smallskip

For convenience, let us first assume that all expected values under
consideration are well defined henceforth. In accordance with the
Kelly criterion, the desired optimal constant strategy corresponds
to the portfolio that maximizes the expected logarithmic return over
$\mathcal{B}^{m}$, that is, $\max_{b\in\mathcal{B}^{m}}\mathbb{E}\big(\log\big<b,X\big>\big)$,
where $X$ represents the random vector of assets' returns. Such an
optimal portfolio is commonly referred to as the \emph{log-optimal
portfolio} in the literature, with respect to the known distribution
of the random vector $X$. Accordingly, Definition \ref{k-log-optimal}
generalizes the concept of the log-optimal portfolio to the case involving
the joint distribution of multiple random vectors.
\begin{defn}
Consider $k\geq2$ random vectors of assets' returns $X^{i}$, $i\in\big\{1,...,k\big\}$,
distributed according to a joint law. Then, $k$-log-optimal portfolios
with respect to that joint distribution are defined as the optimal
tuple of $k$ portfolios, denoted by $\big(b^{1*},...,b^{k*}\big)$,
among all tuples $\big(b^{1},...,b^{k}\big)\in\mathcal{B}^{m\times k}$,
such that the following condition holds:\label{k-log-optimal}
\[
\mathbb{E}\big(\sum_{i=1}^{k}\log\big<b^{i*},X^{i}\big>\big)=\max_{(b^{1},...,b^{k})\in\mathcal{B}^{m\times k}}\mathbb{E}\big(\sum_{i=1}^{k}\log\big<b^{i},X^{i}\big>\big).\smallskip
\]
\end{defn}
\begin{rem*}
The problem of finding the $k$-log-optimal portfolios reduces to
solving $k$ separate maximization problems of expected logarithmic
portfolio returns, i.e., finding $k$ different log-optimal portfolios,
only when the $k$ random vectors of assets' returns $X^{i}$, with
$i\in\big\{1,...,k\big\}$, are independently distributed. In contrast,
as demonstrated in the proof of Theorem \ref{Theorem 1}, the best
$k$-cyclic constant strategy at any time $n\geq k$ can always be
obtained by separately identifying $k$ different best constant strategies,
regardless of the dependence structure among the $k$ variables.\smallskip
\end{rem*}
Let us note that the log-optimal portfolio not only maximizes the
expected logarithmic return but also maximizes the expected portfolio
return relative. In detail, the works of \citet{Bell1980,Algoet1988}
show that the Kuhn-Tucker condition for the optimality of the log-optimal
portfolio, denoted by $b^{*}$, implies the following equivalence:
\[
\mathbb{E}\Big(\log\frac{\big<b,X\big>}{\big<b^{*},X\big>}\Big)\leq0,\,\forall b\in\mathcal{B}^{m}\text{ if and only if }\mathbb{E}\Big(\frac{\big<b,X\big>}{\big<b^{*},X\big>}\Big)\leq1,\,\forall b\in\mathcal{B}^{m}.
\]
This equivalence is essential for demonstrating the asymptotic optimality
in growth rate of the constant strategy $\big(b^{*}\big)$. In a similar
manner, the Kuhn-Tucker condition is extended to the $k$-log-optimal
portfolios in Lemma \ref{Lemma 2}, which affirms that the $k$-log-optimal
portfolios also maximize the expected geometric mean of the $k$ portfolio
return relatives. This result will be subsequently used to characterize
the asymptotic optimality in growth rate of the $k$-cyclic constant
strategy defined corresponding to the $k$-log-optimal portfolios. 
\begin{lem}
\emph{(Kuhn-Tucker condition for the $k$-log-optimal portfolios)}
Consider $k$ random vectors of assets' returns $X^{i}$, $i\in\big\{1,...,k\big\}$.
Then:\label{Lemma 2}
\begin{enumerate}
\item The expectation $\mathbb{E}\big(\sum_{i=1}^{k}\log\big<b^{i},X^{i}\big>\big)$
with respect to any joint distribution of $\big(X^{1},...,X^{k}\big)$
is concave in $\big(b^{1},...,b^{k}\big)$ over $\mathcal{B}^{m\times k}.$
\item The $k$-log-optimal portfolios corresponding to any joint distribution
of $\big(X^{1},...,X^{k}\big)$ satisfy the following necessary condition:
\[
\mathbb{E}\bigg(\prod_{i=1}^{k}\sqrt[k]{\frac{\big<b^{i},X^{i}\big>}{\big<b^{i*},X^{i}\big>}}\bigg)\leq1,\,\forall\big(b^{1},...,b^{k}\big)\in\mathcal{B}^{m\times k}.
\]
\end{enumerate}
\end{lem}
\begin{proof}
Since the function $\log\big<b,X^{i}\big>$ is concave in $b\in\mathcal{B}^{m}$
for each $X^{i}$ with $i\in\left\{ 1,...,k\right\} $, the summation
$\sum_{i=1}^{k}\log\big<b^{i},X^{i}\big>$ is concave in $\big(b^{1},...,b^{k}\big)\in\mathcal{B}^{m\times k}$.
Hence, for any two tuples of portfolios $\big(\hat{b}^{1},...,\hat{b}^{k}\big),\,\big(\bar{b}^{1},...,\bar{b}^{k}\big)\in\mathcal{B}^{m\times k}$,
the following inequality holds:
\[
\mathbb{E}\big(\sum_{i=1}^{k}\log\big<\lambda\hat{b}^{i}+\big(1-\lambda\big)\bar{b}^{i},X^{i}\big>\big)\geq\lambda\mathbb{E}\big(\sum_{i=1}^{k}\log\big<\hat{b}^{i},X^{i}\big>\big)+\big(1-\lambda\big)\mathbb{E}\big(\sum_{i=1}^{k}\log\big<\bar{b}^{i},X^{i}\big>\big),\forall\lambda\in\left[0,1\right],
\]
where the expectation is taken with respect to the joint distribution
of $\big(X^{1},...,X^{k}\big)$. This inequality confirms the concavity
of the function $\mathbb{E}\big(\sum_{i=1}^{k}\log\big<b^{i},X^{i}\big>\big)$
in $\big(b^{1},...,b^{k}\big)\in\mathcal{B}^{m\times k}$, as required
for assertion (a).\smallskip

To prove assertion (b), consider the $k$-log-optimal portfolios $\big(b^{1*},...,b^{k*}\big)$
and any other competing tuple of portfolios $\big(b^{1},...,b^{k}\big)$.
Then, for any $\lambda\in\left[0,1\right]$, define the convex combination:
\[
\big(b^{1,\lambda},...,b^{k,\lambda}\big)\coloneqq\big(1-\ensuremath{\lambda}\big)\big(b^{1*},...,b^{k*}\big)+\lambda\big(b^{1},...,b^{k}\big).
\]
Since the function $\mathbb{E}\big(\sum_{i=1}^{k}\log\big<b^{i},X^{i}\big>\big)$
is concave in $\big(b^{1},...,b^{k}\big)\in\mathcal{B}^{m\times k}$,
according to assertion (a), the Kuhn-Tucker condition implies that
the tuple of portfolios $\big(b^{1*},...,b^{k*}\big)$ is $k$-log-optimal
if and only if the directional derivative of the expected logarithmic
return is non-positive when moving from $\big(b^{1*},...,b^{k*}\big)$
to any other tuple $\big(b^{1},...,b^{k}\big)$ in $\mathcal{B}^{m\times k}$.
Hence, we have:
\begin{align*}
0 & \geq\frac{\dif}{\dif\lambda}\mathbb{E}\big(\sum_{i=1}^{k}\log\big<b^{i,\lambda},X^{i}\big>\big)\Big|_{\lambda=0+}\\
 & =\lim_{\lambda\to0}\frac{1}{\lambda}\mathbb{E}\big(\sum_{i=1}^{k}\log\big<b^{i,\lambda},X^{i}\big>-\sum_{i=1}^{k}\log\big<b^{i,0},X^{i}\big>\big)\\
 & =\lim_{\lambda\to0}\frac{1}{\lambda}\mathbb{E}\Big(\sum_{i=1}^{k}\log\Big(\frac{\big(1-\ensuremath{\lambda}\big)\big<b^{i*},X^{i}\big>+\lambda\big<b^{i},X^{i}\big>}{\big<b^{i*},X^{i}\big>}\Big)\Big)\\
 & =\lim_{\lambda\to0}\frac{1}{\lambda}\mathbb{E}\Big(\sum_{i=1}^{k}\log\Big(1+\lambda\Big(\frac{\big<b^{i},X^{i}\big>}{\big<b^{i*},X^{i}\big>}-1\Big)\Big)\Big)\\
 & =\mathbb{E}\Big(\sum_{i=1}^{k}\frac{\big<b^{i},X^{i}\big>}{\big<b^{i*},X^{i}\big>}\Big)-k,
\end{align*}
where the last equality follows from the limit (\ref{limit interchange})
established in the subproof below. Subsequently, by further applying
the AM-GM inequality, which holds pointwise for any outcome of the
$k$ variables $\big(X^{1},...,X^{k}\big)$, the above inequality
yields:
\[
\mathbb{E}\bigg(\prod_{i=1}^{k}\sqrt[k]{\frac{\big<b^{i},X^{i}\big>}{\big<b^{i*},X^{i}\big>}}\bigg)\leq\frac{1}{k}\mathbb{E}\Big(\sum_{i=1}^{k}\frac{\big<b^{i},X^{i}\big>}{\big<b^{i*},X^{i}\big>}\Big)\leq1.
\]
This is the necessary condition for the $k$-log-optimal portfolios
$\big(b^{1*},...,b^{k*}\big)$.\smallskip

(\emph{Subproof for the limit of expected value}). Recall the basic
upper bound $\log\big(1+Z\big)\leq Z$ for $Z>-1$. Similarly, by
establishing the upper bound for $\log\big(1-Z/\big(1+Z\big)\big)$
with $Z>0$, we obtain the following lower bound for $\log\big(1+Z\big)$:
\[
\log\big(1+Z\big)=-\log\Big(1-\frac{Z}{1+Z}\Big)\geq\frac{Z}{1+Z}.
\]
Hence, defining $U\coloneqq\big(\big<b^{i},X^{i}\big>/\big<b^{i*},X^{i}\big>-1\big)\geq-1$
for $i\in\big\{1,...,k\big\}$ and $\lambda>0$, and applying the
established bounds above, we obtain the following inequalities:
\[
\frac{U}{1+\lambda U}\leq\frac{1}{\lambda}\log\big(1+\lambda U\big)\leq U.
\]
Moreover, since $|U/\big(1+\lambda U\big)|\leq2|U|$ for all $\lambda\leq1/2$
and $\mathbb{E}\big(U\big)<\infty$, it follows from the Lebesgue
dominated convergence theorem, together with the inequalities established
above, that:
\begin{equation}
\lim_{\lambda\to0}\frac{1}{\lambda}\mathbb{E}\big(\log\big(1+\lambda U\big)\big)=\mathbb{\mathbb{E}}\Big(\lim_{\lambda\to0}\frac{U}{1+\lambda U}\Big)=\mathbb{E}\big(U\big),\label{limit interchange}
\end{equation}
which completes the proof of the assertion.
\end{proof}
Since the classical Kelly criterion and its associated log-optimal
portfolio are valid as long as the stochastic market is i.i.d., its
generalized form, namely the $k$-cyclic constant strategy defined
by the $k$-log-optimal portfolios, is expected to attain the optimal
asymptotic growth rate compared to all dynamic strategies in a particular
stochastic market where each collection of $k$ assets' returns variables
is distributed according to an identical joint law. This implies that
the market exhibits correlations between assets' returns at different
periods within a fixed window length or contains hidden patterns that
cause certain typical tendencies in volatility over time. Specifically,
the desired asymptotic optimality is proven under a block-wise i.i.d.
stochastic process, that is, a process consisting of i.i.d. random
matrices, in Theorem \ref{Theorem 2}, thereby extending the classical
Kelly criterion to markets exhibiting serial dependence. Essentially,
a block-wise i.i.d. process is a non-stationary process in the conventional
temporal sense, meaning that it is not independent across time in
the usual way. Instead, it displays a periodic structure with internal
dependencies among the assets' returns within each time block, which
can be viewed as a matrix. The case of persistent day-of-the-week
effects in Example \ref{Example k wise strategy} can be regarded
as an instance of serial dependence within a time block. There exists
a vast and long standing body of literature in finance and economics
on such calendar effects and patterns, including works by \citet{Cross1973,French1980,Harris1986,Pearce1996,Zilca2017,Grebe2024},
among many others.
\begin{thm}
\emph{(Generalized Kelly criterion and asymptotic optimality)} Assume
that the process of assets' returns $\big\{ X_{n}\big\}_{n=1}^{\infty}$
is represented by the block-wise i.i.d. process $\big\{\big(X_{kt+1},X_{kt+2},...,X_{kt+k}\big)\big\}_{t=0}^{\infty}$
according to a joint distribution of the random variables $\big(X_{1},X_{2},...,X_{k}\big)$,
with $k\geq2$. Consider the $k$-cyclic constant strategy $\big(b_{n}^{k*}\big)$,
defined by the $k$-log-optimal portfolios $\big(b^{1*},...,b^{k*}\big)$
with respect to this joint distribution, and any other competing strategy
$\big(b_{n}\big)$. Then, the following holds:\label{Theorem 2} 
\[
\limsup_{t\to\infty}\big(W_{kt+k}\big(b_{kt+k}\big)-W_{kt+k}\big(b_{kt+k}^{k*}\big)\big)\leq0,\text{ \emph{a.s.}}
\]
Consequently, the highest asymptotic growth rate attainable by a dynamic
strategy is:
\[
\lim_{t\to\infty}W_{kt+k}\big(b_{kt+k}^{k*}\big)=\max_{(b^{1},...,b^{k})\in\mathcal{B}^{m\times k}}\frac{1}{k}\mathbb{E}\big(\sum_{i=1}^{k}\log\big<b^{i},X_{i}\big>\big),\text{ \emph{a.s.}}
\]
\end{thm}
\begin{proof}
Since the block-wise process $\big\{\big(X_{kt+1},X_{kt+2},...,X_{kt+k}\big)\big\}_{t=0}^{\infty}$
is i.i.d., applying Lemma \ref{Lemma 2} results in the following
inequality:
\[
\mathbb{E}\bigg(\sqrt[k]{\frac{S_{kt+k}\big(b_{kt+k}\big)}{S_{kt+k}\big(b_{kt+k}^{k*}\big)}}\bigg)=\prod_{j=0}^{t}\mathbb{E}\bigg(\prod_{i=1}^{k}\sqrt[k]{\frac{\big<b_{kj+i},X_{kj+i}\big>}{\big<b^{i*},X_{kj+i}\big>}}\bigg)\leq1,\text{ }\forall t\geq0.
\]
Then, by applying the Markov inequality, we further have that:
\begin{align*}
\mathbb{P}\bigg(\sqrt[k]{\frac{S_{kt+k}\big(b_{kt+k}\big)}{S_{kt+k}\big(b_{kt+k}^{k*}\big)}}>t^{2}\bigg) & \leq\frac{1}{t^{2}}\\
\Rightarrow\,\,\,\mathbb{P}\Big(\frac{1}{k\big(t+1\big)}\log\frac{S_{kt+k}\big(b_{kt+k}\big)}{S_{kt+k}\big(b_{kt+k}^{k*}\big)}>\frac{2}{t+1}\log t\Big) & \leq\frac{1}{t^{2}}\\
\Rightarrow\sum_{t=1}^{\infty}\mathbb{P}\Big(\frac{1}{kt+k}\log\frac{S_{kt+k}\big(b_{kt+k}\big)}{S_{kt+k}\big(b_{kt+k}^{k*}\big)}>\frac{2}{t+1}\log t\Big) & \leq\sum_{t=1}^{\infty}\frac{1}{t^{2}}=\frac{\pi^{2}}{6}.
\end{align*}
Here, the last equality follows from the result of the Basel problem.
Finally, by invoking the Borel-Cantelli lemma, we obtain the following:
\[
\mathbb{P}\Big(\frac{1}{kt+k}\log\frac{S_{kt+k}\big(b_{kt+k}\big)}{S_{kt+k}\big(b_{kt+k}^{k*}\big)}>\frac{2}{t+1}\log t,\text{ infinitely often}\Big)=0,
\]
which is equivalent to the desired inequality:
\[
\limsup_{t\to\infty}\big(W_{kt+k}\big(b_{kt+k}\big)-W_{kt+k}\big(b_{kt+k}^{k*}\big)\big)=\limsup_{t\to\infty}\frac{1}{kt+k}\log\frac{S_{kt+k}\big(b_{kt+k}\big)}{S_{kt+k}\big(b_{kt+k}^{k*}\big)}\leq0,\text{ a.s}.
\]
This implies that $\big(b_{n}^{k*}\big)$ sets the highest asymptotic
growth rate among all strategies.\smallskip

Meanwhile, to demonstrate the limit of the asymptotic growth rate
of the strategy $\big(b_{n}^{k*}\big)$, first note that the sequence
$\big\{\big(\big<b^{1*},X_{kt+1}\big>,...,\big<b^{k*},X_{kt+k}\big>\big)\big\}_{t=0}^{\infty}$
is also a block-wise i.i.d. process. Moreover, this block-wise i.i.d.
process can be equivalently represented as a standard i.i.d. process
$\big\{ Z_{n}\big\}_{n=1}^{\infty}$, where $Z_{n}\coloneqq\big(Z_{n,1},...,Z_{n,k}\big)\in\mathbb{R}_{++}^{k}$
with $Z_{n,i}=\big<b^{i*},X_{k(n-1)+i}\big>$ for any $n\geq1$ and
$i\in\big\{1,...,k\big\}$. Then, invoking the law of large numbers
yields the desired limit:
\begin{align*}
\lim_{t\to\infty}W_{kt+k}\big(b_{kt+k}^{k*}\big) & =\lim_{t\to\infty}\frac{1}{k\big(t+1\big)}\sum_{j=0}^{t}\big(\sum_{i=1}^{k}\log\big<b^{i*},X_{kj+i}\big>\big)\\
 & =\frac{1}{k}\mathbb{E}\big(\sum_{i=1}^{k}\log\big<b^{i*},X_{i}\big>\big),\text{ a.s.}\text{ }\\
 & =\max_{(b^{1},...,b^{k})\in\mathcal{B}^{m\times k}}\frac{1}{k}\mathbb{E}\big(\sum_{i=1}^{k}\log\big<b^{i},X_{i}\big>\big),\text{ a.s.,}
\end{align*}
due to the fact that $\big(b^{1*},...,b^{k*}\big)$ are the $k$-log-optimal
portfolios. This completes the proof.
\end{proof}
Let us remark on a subtle difference in asymptotic optimality between
the constant strategy defined by the log-optimal portfolio in the
classical Kelly criterion and the $k$-cyclic constant strategy defined
by the $k$-log-optimal portfolios in the Generalized Kelly criterion.
According to Theorem \ref{Theorem 2} below, $W_{kt+k}\big(b_{kt+k}\big)<W_{kt+k}\big(b_{kt+k}^{k*}\big)+\epsilon$
eventually as $t$ increases to infinity, for any arbitrarily small
$e>0$. This implies that the growth rate of the competing strategy
$\big(b_{n}\big)$ cannot exceed the growth rate of the best $k$-cyclic
constant strategy $\big(b_{n}^{k*}\big)$ by a margin of $\epsilon$
infinitely often, and this is evaluated at periodic steps of size
$k$, rather than the standard unit step as in the case $k=1$, which
corresponds to the classical Kelly criterion. Therefore, the established
optimum rate $\max_{(b^{1},...,b^{k})\in\mathcal{B}^{m\times k}}k^{-1}\mathbb{E}\big(\sum_{i=1}^{k}\log\big<b^{i},X_{i}\big>\big)$
in the theorem does not necessarily mean that the growth rate of the
best one among all dynamic strategies must converge, as is implied
by the asymptotic optimality for the i.i.d. market; rather, the growth
rate must converge to the optimum rate only at the observation time
points $kt+k$ as $t$ increases.\smallskip

In Theorem \ref{Theorem 2}, if the joint distribution of the $k$
random variables of assets' returns $\big(X_{1},...,X_{k}\big)$ is
unknown, then the $k$-cyclic constant strategy corresponding to the
$k$-log-optimal portfolios is also unidentified. However, since the
learning algorithm in (\ref{proposal strategy}) relies only on the
observed realizations of assets' returns, the resulting strategy $\big(b_{n}\big)$
achieves the optimal asymptotic growth rate without requiring knowledge
of the true joint distribution, as demonstrated in Corollary \ref{corollary 2}
below. Moreover, Corollary \ref{corollary 1} provides a useful application
in the case where the actual value of $k$, representing the number
of random variables $\big(X_{1},...,X_{k}\big)$, is also unknown.
As long as the true $k$ lies within a known small set of possibilities
$\big\{ k^{1},...,k^{z}\big\}$, the constructed strategy $\big(b_{n}\big)$,
using the product $\zeta\coloneqq\prod_{i=1}^{z}k^{i}$ as the cycle
length, will still achieve the optimal asymptotic growth rate, similarly
to the $k$-cyclic constant strategy corresponding to the true $k$-log-optimal
portfolios. In particular, if the stochastic market is simply i.i.d.,
then the constructed strategy $\big(b_{n}\big)$, for any $k\geq2$,
yields the same highest asymptotic growth rate as the optimal constant
strategy.\smallskip
\begin{rem*}
In the literature, there exist learning algorithms for constructing
strategies that can attain the optimal asymptotic growth rate in a
stationary ergodic market process $\big\{ X_{n}\big\}_{n=1}^{\infty}$,
without requiring knowledge of the distribution, as proposed in \citet{Algoet1992,Gyorfi2006}.
However, since the block-wise i.i.d. stochastic process $\big\{\big(X_{kt+1},...,X_{kt+k}\big)\big\}_{t=0}^{\infty}$
is non-stationary, these learning algorithms do not apply, whereas
the learning algorithm proposed in (\ref{proposal strategy}) does.
\end{rem*}
\begin{cor}
Under the block-wise i.i.d. process $\big\{\big(X_{kt+1},...,X_{kt+k}\big)\big\}_{t=0}^{\infty}$
of assets' returns, with $k\geq2$, the strategy $\big(b_{n}\big)$
constructed according to (\ref{proposal strategy}), with cycle length
$dk$ for any $d\geq1$, attains the same optimal asymptotic growth
rate as the optimal $k$-cyclic constant strategy $\big(b_{n}^{k*}\big)$,
i.e.,\label{corollary 2}
\[
\lim_{t\to\infty}\big(W_{kt+k}\big(b_{kt+k}\big)-W_{kt+k}\big(b_{kt+k}^{k*}\big)\big)=0,\text{ \emph{a.s.}}
\]
\end{cor}
\begin{proof}
By the definition of the best $k$-cyclic constant strategy in hindsight
over time, and the asymptotic optimality of the $k$-cyclic constant
strategy $\big(b_{n}^{k*}\big)$ according to Theorem \ref{Theorem 2},
we have: 
\begin{align*}
0\leq & \liminf_{t\to\infty}\max_{(b^{1},...,b^{k})\in\mathcal{B}^{m\times k}}\big(W_{kt+k}\big(b_{kt+k}^{k}\big)-W_{kt+k}\big(b_{kt+k}^{k*}\big)\big)\\
\leq & \limsup_{t\to\infty}\max_{(b^{1},...,b^{k})\in\mathcal{B}^{m\times k}}\big(W_{kt+k}\big(b_{kt+k}^{k}\big)-W_{kt+k}\big(b_{kt+k}^{k*}\big)\big)\leq0,\text{ a.s.}
\end{align*}
Then, Corollary \ref{corollary 1} implies the desired limit for the
growth rate of the strategy $\big(b_{n}\big)$ constructed with cycle
length $dk$, for any $d\geq1$, and thus completes the proof.
\end{proof}

\section{Numerical experiments}

In this section, numerical experiments are conducted to investigate
the empirical performance of the proposed strategy, constructed according
to (\ref{proposal strategy}), on a real market dataset. This setting
departs significantly from the worst-case scenario of assets' returns
considered in the consistency guarantee of Theorem \ref{Theorem 1}.
The investigation examines the impact of the cycle length on the strategy\textquoteright s
learning ability and its capacity to outperform the best constant
strategy in terms of cumulative wealth. Moreover, the experiments
highlight the strong potential of the best $k$-cyclic constant strategies
in hindsight, justifying their benchmarking role and demonstrating
the benefit of leveraging differences in the characteristics of subsequences
of assets' returns.\smallskip

\textbf{Testing strategies, abbreviations, and computation}. The experiments
are conducted solely with strategies constructed according to the
proposed learning algorithm (\ref{proposal strategy}) using uniform
$\mu(\cdot)$ for all cycle lengths $k\in\big\{1,...,10\big\}$. These
strategies are referred to as $k$\emph{-parallel Universal Portfolio}
strategies ($k$-PUP) and are denoted accordingly by $\big(b_{n}^{k\text{-PUP}}\big)$
throughout this section. It is important to recall that the $1$-PUP
strategy $\big(b_{n}^{1\text{-PUP}}\big)$ coincides exactly with
the Universal Portfolio strategy, and the associated $1$-cyclic constant
strategy is simply the constant strategy, as noted in Definition \ref{k-wise strategy definition}.
Therefore, for convenience and to avoid redundant notation, $\big(b_{n}^{1\text{-PUP}}\big)$
is used to denote the Universal Portfolio strategy, while the $k$-cyclic
constant strategy ($k$-CC), with associated notation $\big(b_{n}^{k}\big)$,
is also used to refer to the constant strategy in the case $k=1$
henceforth. The $k$-PUP strategies are numerically approximated via
Riemann summation over a discretization of the simplex $\mathcal{B}^{m}$
for computing the integrals in (\ref{proposal strategy}). Thus, while
finer discretization grids lead to more accurate approximations, they
also entail significantly higher computational complexity, especially
in high-dimensional simplices, i.e., with a large number of stocks
in the portfolio.\smallskip

\textbf{Dataset and stocks}. To facilitate the investigation of the
asymptotic properties of the strategies of interest, the dataset used
in the experiments covers the daily adjusted closing prices in USD
over a long time horizon of $27$ years, from December 31, 1992, to
December 31, 2019, i.e., $6798$ trading days, provided by the Center
for Research in Security Prices, LLC (CRSP). To reduce the computational
complexity associated with the discretization of the high-dimensional
simplex $\mathcal{B}^{m}$, a small number $m=4$ of blue-chip companies
listed on the NYSE and NASDAQ exchanges are selected. These span the
industries of semiconductors, technology, aerospace, automation, defense,
and financial services: Honeywell International Inc. (HON), The Boeing
Company (BA), Advanced Micro Devices, Inc. (AMD), and JPMorgan Chase
\& Co. (JPM). Figure \ref{figure 1} depicts the evolution of these
stocks over the $27$-year period, with observable impacts from several
critical market events, such as the 1994 bond market crash, the Asian
financial crisis, the dot-com crash, the 2008 global financial crisis,
and the pre-COVID-19 period.
\begin{figure}[H]
\begin{centering}
\includegraphics[scale=0.33]{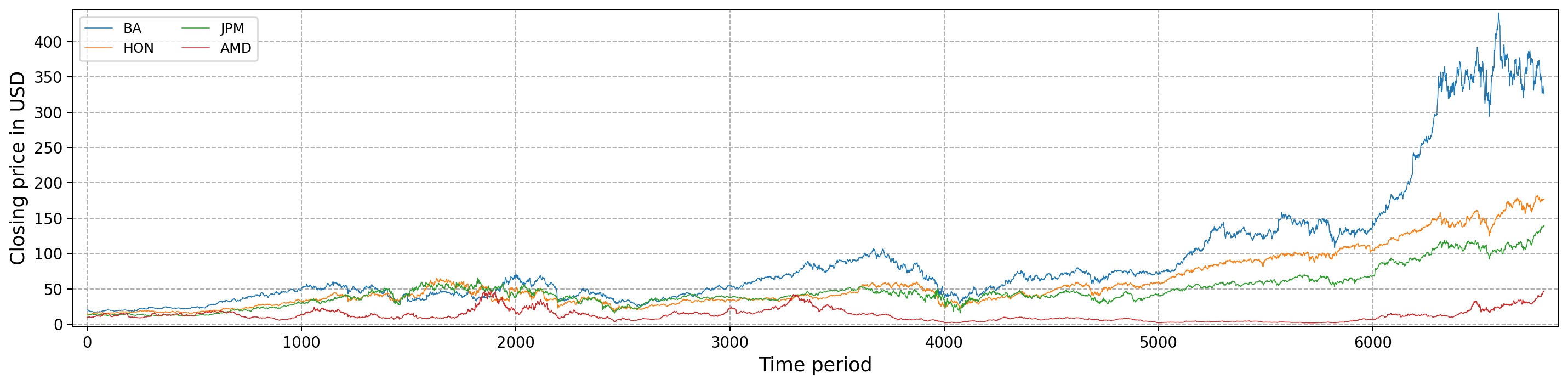}
\par\end{centering}
\caption{Adjusted closing prices in USD of the stocks from December 31, 1992,
to December 31, 2019.\label{figure 1}}
\end{figure}
\smallskip

\textbf{Performance measurement}. The performance of the involved
strategies is primarily presented through several figures to clearly
illustrate their evolution in growth rate and cumulative wealth over
a long-term investment horizon, which are fundamental aspects of the
proposed learning algorithm. Additionally, some other common performance
metrics, including average returns and Sharpe ratios, are also reported
to further assess the impact of the learning algorithm\textquoteright s
decomposition of the assets' returns sequence.

\subsection{Results and analysis on the impact of parameter $\bm{k}$ tuning}

The experimental results affirm the capacity of the $k$-PUP strategies
to eventually outperform the best constant strategy, which serves
as the baseline for comparison, in terms of cumulative wealth. This
improvement benefits from tuning an appropriately higher cycle length
$k$ in the learning algorithm. Specifically, the left graphic of
Figure \ref{figure 2} shows that the cumulative wealth of the best
constant strategy, i.e., the best $1$-CC strategy in hindsight at
the final day $n=6798$, is exceeded by those of the $2$-PUP and
$6$-PUP strategies, with the former significantly outperforming the
latter. Meanwhile, the $4$-PUP and $9$-PUP strategies yield slightly
lower cumulative wealth than the baseline, and the $1$-PUP performs
the worst, with a notably lower value. This outcome aligns with the
strict inequality established in (\ref{UP bound}) for the Universal
Portfolio strategy. Additionally, the performance metrics reported
in Table \ref{table 1} indicate that a $k$-PUP strategy achieving
a higher growth rate also tends to yield a higher average return,
though not necessarily a higher Sharpe ratio. This tendency suggests
that the $k$-PUP learning algorithm capitalizes on differences in
average assets' returns across the decomposed subsequences.\smallskip

In general, the results above indicate that even values of the cycle
length $k$ are better choices, although increasing $k$ does not
necessarily lead to higher cumulative wealth. This behavior is influenced
by the differences in cumulative wealth achieved by the best benchmarking
$k$-CC strategies associated with the corresponding $k$-PUP strategies,
and by their upper bounds established in Theorem \ref{Theorem 1}.
In detail, the right graphic of Figure \ref{figure 2} illustrates
that all the best $k$-CC strategies with $k\geq2$ yield much higher
final cumulative wealth than the best constant strategy, with some
cases achieving up to nearly $80$ times more\footnote{The huge gap in cumulative wealth, up to thousands, between the best
$k$-CC strategies in hindsight is not a surprise, as they result
from the product of cumulative wealths yielded by multiple best constant
strategies defined on $k$ decomposed subsequences of assets' returns,
as shown in the proof of Theorem \ref{Theorem 1}. This implies that
the best $k$-CC strategies would be exceptionally profitable in real-world
trading, yet unfortunately, they can only be known in hindsight.}. Moreover, the final cumulative wealth associated with even values
of $k$ generally increases with $k$, whereas this trend does not
hold for odd values, which is explained by Corollary \ref{corollary 1}.
Since the upper bound for the difference $\max_{(b^{1},...,b^{k})\in\mathcal{B}^{m\times k}}\big(\log S_{n}\big(b_{n}^{k}\big)-\log S_{n}\big(b_{n}^{k\text{-PUP}}\big)\big)$
established in Theorem \ref{Theorem 1} increases with $k$, a $k$-PUP
strategy with large $k$ requires a longer period to close the gap
in growth rate compared to the associated best $k$-CC strategy. Therefore,
when using a small cycle length such as $k=2$, which results in a
sufficiently large number of cycles relative to the total investment
period, we obtain $S_{6798}\big(b_{6798}^{2\text{-PUP}}\big)>\max_{b\in\mathcal{B}^{m}}S_{6798}\big(b\big)$,
as $\max_{(b^{1},b^{2})\in\mathcal{B}^{m\times2}}S_{6798}\big(b_{6798}^{2}\big)\gg\max_{b\in\mathcal{B}^{m}}S_{6798}\big(b\big)$,
consistent with the reported results. Meanwhile, although a larger
cycle length such as $k=6$ implies a slower convergence due to a
greater gap between the $6$-PUP strategy and the best $6$-CC strategy,
the result $S_{6798}\big(b_{6798}^{6\text{-PUP}}\big)>$$\max_{b\in\mathcal{B}^{m}}S_{6798}\big(b\big)$,
albeit marginally, can be attributed to the fact that the extremely
high value of $\max_{(b^{1},...,b^{6})\in\mathcal{B}^{m\times6}}S_{6798}\big(b_{6798}^{6}\big)$
compared to $\max_{b\in\mathcal{B}^{m}}S_{6798}\big(b\big)$ is sufficient
to compensate for this gap.
\begin{figure}[H]
\begin{centering}
\includegraphics[scale=0.33]{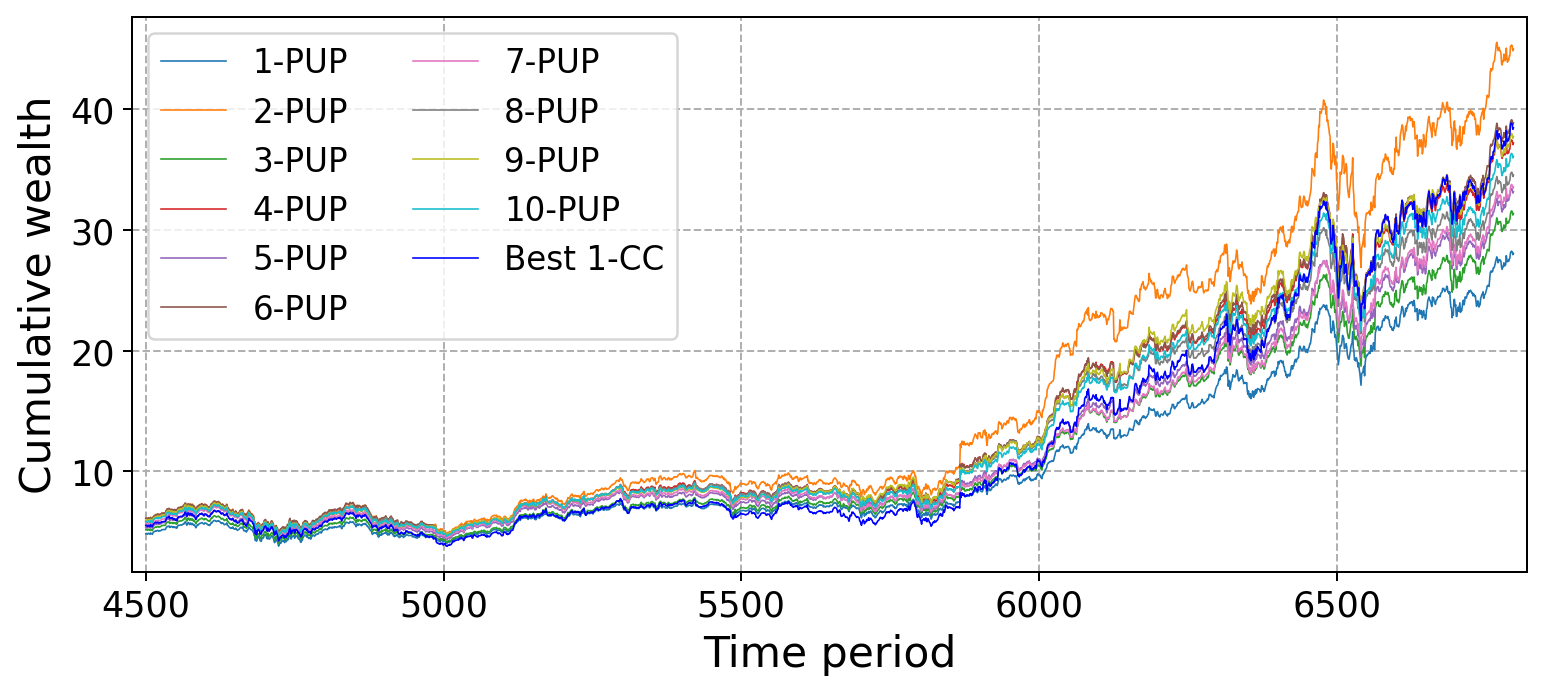}~\includegraphics[scale=0.33]{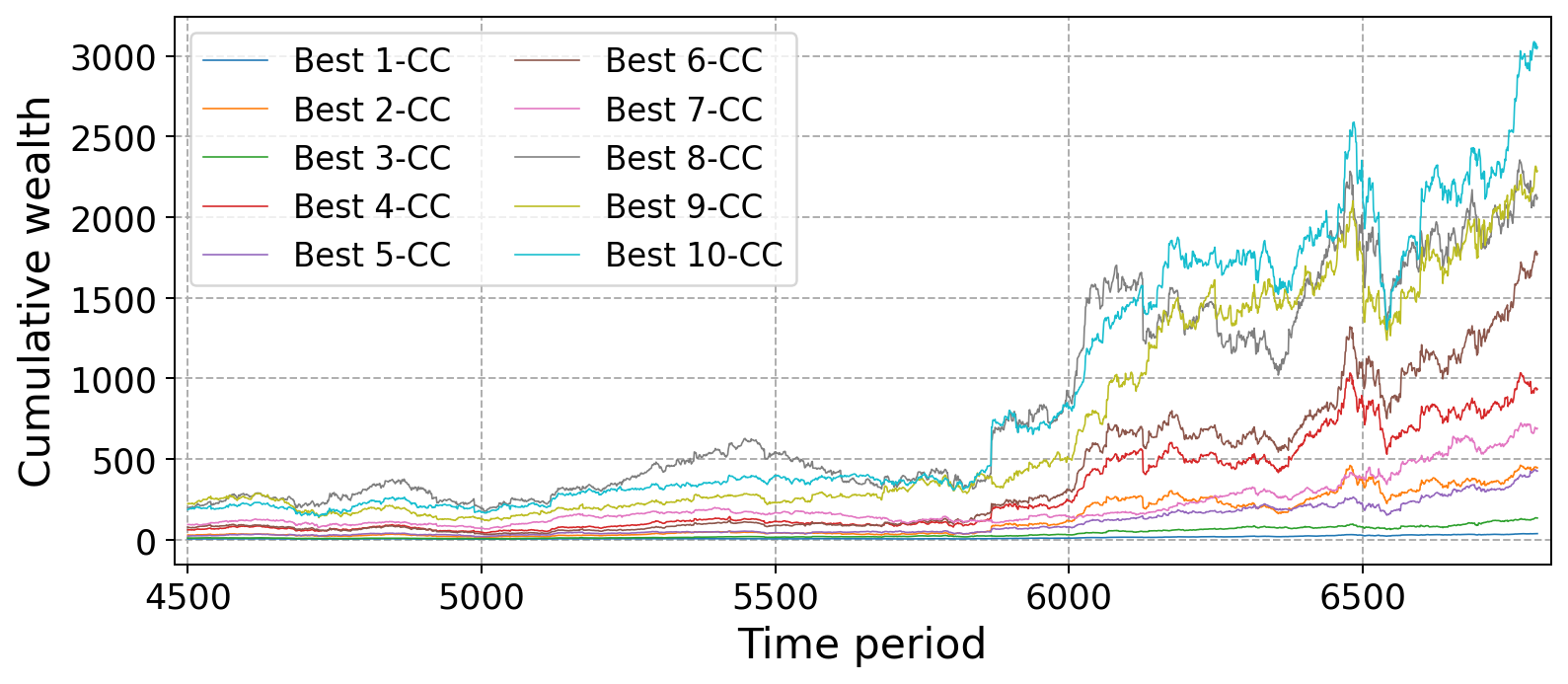}
\par\end{centering}
\caption{Excerpt showing the evolution of the strategies over the final $2299$
days of the investment period, with the best $k$-CC strategies are
determined in hindsight after the end date.\label{figure 2}}
\end{figure}
\medskip

\renewcommand{\arraystretch}{0.3}
\begin{table}[H]
\caption{Performance measures of strategies.\label{table 1}}

\begin{centering}
\begin{tabular*}{12cm}{@{\extracolsep{\fill}}>{\raggedright}m{2.5cm}>{\raggedright}m{2cm}>{\raggedright}m{2cm}>{\raggedright}m{2cm}>{\raggedright}m{2cm}}
\toprule 
{\tiny\textbf{Strategy}} & {\tiny\textbf{Final wealth}} & {\tiny\textbf{Growth rate}} & {\tiny\textbf{Average return}} & {\tiny\textbf{Sharpe ratio}}\tabularnewline
\bottomrule
\end{tabular*}
\par\end{centering}
\begin{centering}
\begin{tabular*}{12cm}{@{\extracolsep{\fill}}>{\raggedright}p{2.5cm}>{\raggedright}p{2cm}>{\raggedright}p{2cm}>{\raggedright}p{2cm}>{\raggedright}p{2cm}}
\toprule 
\multicolumn{5}{c}{{\tiny Referential strategies (not shown in Figure \ref{figure 2})}}\tabularnewline
\midrule
{\tiny Buy and hold on BA} & {\tiny 16.237259} & {\tiny 0.000410} & {\tiny 1.000587} & {\tiny 53.250233}\tabularnewline
{\tiny Buy and hold on HON} & {\tiny 12.261668} & {\tiny 0.000369} & {\tiny 1.000546} & {\tiny 53.153077}\tabularnewline
{\tiny Buy and hold on JPM} & {\tiny 10.827184} & {\tiny 0.000350} & {\tiny 1.000618} & {\tiny 43.049489}\tabularnewline
{\tiny Buy and hold on AMD} & {\tiny 5.060413} & {\tiny 0.000239} & {\tiny 1.001008} & {\tiny 25.447489}\tabularnewline
\end{tabular*}
\par\end{centering}
\begin{centering}
\begin{tabular*}{12cm}{@{\extracolsep{\fill}}>{\raggedright}p{2.5cm}>{\raggedright}p{2cm}>{\raggedright}p{2cm}>{\raggedright}p{2cm}>{\raggedright}p{2cm}}
\toprule 
\multicolumn{5}{c}{{\tiny Benchmarking $k$-CC strategies}}\tabularnewline
\midrule
{\tiny Best $1$-CC} & {\tiny 38.458466} & {\tiny 0.000537} & {\tiny 1.000737} & {\tiny 49.971121}\tabularnewline
{\tiny Best $2$-CC} & {\tiny 446.723698} & {\tiny 0.000898} & {\tiny 1.001370} & {\tiny 32.270519}\tabularnewline
{\tiny Best $3$-CC} & {\tiny 135.029046} & {\tiny 0.000722} & {\tiny 1.001049} & {\tiny 38.934002}\tabularnewline
{\tiny Best $4$-CC} & {\tiny 934.229610} & {\tiny 0.001006} & {\tiny 1.001423} & {\tiny 34.481101}\tabularnewline
{\tiny Best $5$-CC} & {\tiny 426.248099} & {\tiny 0.000891} & {\tiny 1.001243} & {\tiny 37.573436}\tabularnewline
{\tiny Best $6$-CC} & {\tiny 1773.677037} & {\tiny 0.001100} & {\tiny 1.001553} & {\tiny 33.096868}\tabularnewline
{\tiny Best $7$-CC} & {\tiny 690.315751} & {\tiny 0.000962} & {\tiny 1.001251} & {\tiny 41.492901}\tabularnewline
{\tiny Best $8$-CC} & {\tiny 2121.042123} & {\tiny 0.001127} & {\tiny 1.001555} & {\tiny 33.903400}\tabularnewline
{\tiny Best $9$-CC} & {\tiny 2286.811633} & {\tiny 0.001138} & {\tiny 1.001503} & {\tiny 36.886932}\tabularnewline
{\tiny Best $10$-CC} & {\tiny 3054.244982} & {\tiny 0.001180} & {\tiny 1.001515} & {\tiny 38.282879}\tabularnewline
\end{tabular*}
\par\end{centering}
\begin{centering}
\begin{tabular*}{12cm}{@{\extracolsep{\fill}}>{\raggedright}p{2.5cm}>{\raggedright}p{2cm}>{\raggedright}p{2cm}>{\raggedright}p{2cm}>{\raggedright}p{2cm}}
\toprule 
\multicolumn{5}{c}{{\tiny Evaluated $k$-PUP strategies}}\tabularnewline
\midrule
{\tiny$1$-PUP} & {\tiny 28.052022} & {\tiny 0.000490} & {\tiny 1.000651} & {\tiny 55.915769}\tabularnewline
{\tiny$2$-PUP} & {\tiny 44.980156} & {\tiny 0.000560} & {\tiny 1.000729} & {\tiny 54.459225}\tabularnewline
{\tiny 3-PUP} & {\tiny 31.343436} & {\tiny 0.000507} & {\tiny 1.000671} & {\tiny 55.232045}\tabularnewline
{\tiny$4$-PUP} & {\tiny 37.199659} & {\tiny 0.000532} & {\tiny 1.000698} & {\tiny 54.932296}\tabularnewline
{\tiny$5$-PUP} & {\tiny 33.165428} & {\tiny 0.000515} & {\tiny 1.000680} & {\tiny 55.086373}\tabularnewline
{\tiny$6$-PUP} & {\tiny 38.860547} & {\tiny 0.000538} & {\tiny 1.000704} & {\tiny 54.981575}\tabularnewline
{\tiny$7$-PUP} & {\tiny 33.565030} & {\tiny 0.000517} & {\tiny 1.000682} & {\tiny 55.150411}\tabularnewline
{\tiny$8$-PUP} & {\tiny 34.535888} & {\tiny 0.000521} & {\tiny 1.000689} & {\tiny 54.683205}\tabularnewline
{\tiny$9$-PUP} & {\tiny 37.725267} & {\tiny 0.000534} & {\tiny 1.000700} & {\tiny 55.031850}\tabularnewline
{\tiny$10$-PUP} & {\tiny 36.091242} & {\tiny 0.000528} & {\tiny 1.000693} & {\tiny 55.002617}\tabularnewline
\bottomrule
\end{tabular*}
\par\end{centering}
\centering{}%
\begin{minipage}[t]{12cm}%
\begin{spacing}{0.5}
{\tiny\textbf{Note}}{\tiny . The $k$-PUP and $k$-CC strategies are
numerically approximated using $11437$ discretization points with
step size of $0.025$ per simplex $\mathcal{B}^{4}$.}
\end{spacing}
\end{minipage}
\end{table}
\smallskip

\textbf{Convergence of the strategies' growth rates}. To further clarify
the convergence speed of the growth rate of the $k$-PUP strategy
to that of the corresponding best $k$-CC strategy, the left graphic
of Figure \ref{figure 3} illustrates a comparison of the empirical
differences in growth rates over time for the strategies under consideration.
It reflects our expectation that the empirical convergence to zero
of the difference $\max_{(b^{1},...,b^{k})\in\mathcal{B}^{m\times k}}\big(W_{n}\big(b_{n}^{k}\big)-W_{n}\big(b_{n}^{k\text{-PUP}}\big)\big)$
becomes slower as the cycle length $k$ increases, since the upper
bound established in Theorem \ref{Theorem 1} is larger. However,
due to Corollary \ref{corollary 1}, this assertion comes with some
subtle nuance, as the guarantee holds only when the length increases
specifically in multiples of $k$. For instance, the plot clearly
shows that the convergence speed to zero of the growth rate difference
for $k=2$ is slower than that for $k=3$, and a similar convergence
pattern is also observed for the pair $k=4$ and $k=5$.\smallskip

Moreover, although the $1$-PUP strategy yields significantly worse
cumulative wealth over time compared to all other $k$-PUP strategies,
it is important to examine whether their growth rates converge asymptotically.
Recall that, according to Corollary \ref{corollary 1}, the asymptotic
growth rate of the $1$-PUP strategy should serve as a lower bound
for all other $k$-PUP strategies; thus, it is considered a common
reference for comparison. The right graphic of Figure \ref{figure 3}
illustrates that the difference in growth rate $W_{n}\big(b_{n}^{k\text{-PUP}}\big)-W_{n}\big(b_{n}^{1\text{-PUP}}\big)$
does not converge to zero over the entire investment period for any
$k>1$ (note that $W_{n}\big(b_{n}^{k\text{-PUP}}\big)-W_{n}\big(b_{n}^{1\text{-PUP}}\big)$
is much smaller than $\max_{(b^{1},...,b^{k})\in\mathcal{B}^{m\times k}}\big(W_{n}\big(b_{n}^{k}\big)-W_{n}\big(b_{n}^{k\text{-PUP}}\big)\big)$,
since the cumulative wealth $\max_{(b^{1},...,b^{k})\in\mathcal{B}^{m\times k}}S_{n}\big(b_{n}^{k}\big)$
is several times greater than $S_{n}\big(b_{n}^{k\text{-PUP}}\big)$
for all $k$, as shown in Figure \ref{figure 2}). Therefore, the
assets' returns are unlikely to be i.i.d., and the classical Kelly
criterion does not apply in this context, as the growth rate of the
$1$-PUP strategy does not converge to that of the optimal strategy\footnote{To account for the computational error caused by numerical approximation,
let us assume that the process of assets' returns $\big\{ X_{n}\big\}_{n=1}^{\infty}$
is i.i.d., and define $\Delta\coloneqq\max_{b\in\mathcal{B}^{4}}\mathbb{E}\big(\log\big<b,X_{1}\big>\big)-\lim_{n\to\infty}W_{n}\big(b_{n}^{1\text{-PUP}}\big)$,
which represents the error in asymptotic growth rate between the true
optimal constant strategy and the numerically approximated $1$-PUP
strategy, and is very small. Then, all differences $W_{n}\big(b_{n}^{k\text{-PUP}}\big)-W_{n}\big(b_{n}^{1\text{-PUP}}\big)$
for $k\geq2$ must converge to $\Delta$. However, the evolutions
illustrated in the right graphic of Figure \ref{figure 3} also do
not support this hypothesis.}. As an additional remark, the $k$-PUP strategies with $k\in\big\{4,6,8,10\big\}$
grow asymptotically at lower rates than the $2$-PUP strategy, although
they theoretically should not, due to Corollary \ref{corollary 2},
because their convergence speeds to the corresponding best $k$-CC
strategies are slower, as previously discussed. 
\begin{figure}[H]
\begin{centering}
\includegraphics[scale=0.335]{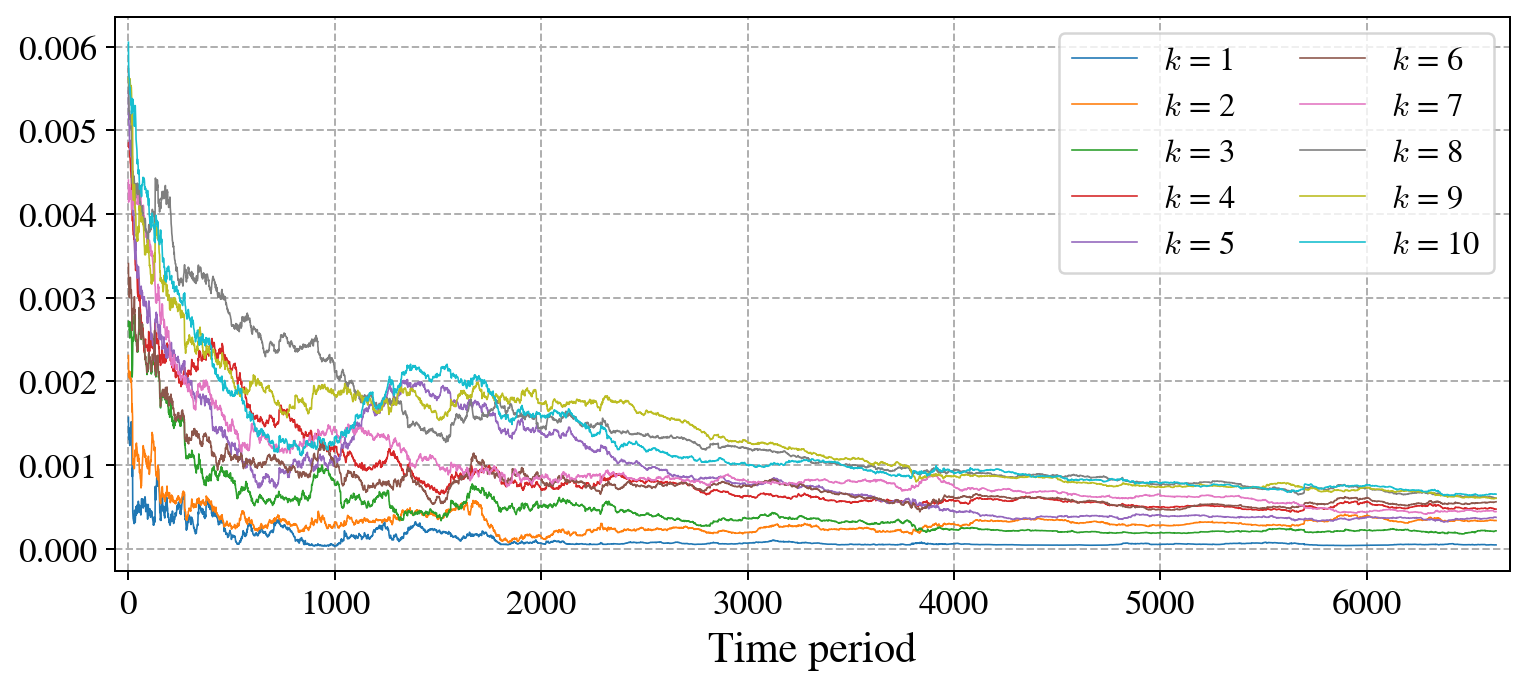}\includegraphics[scale=0.335]{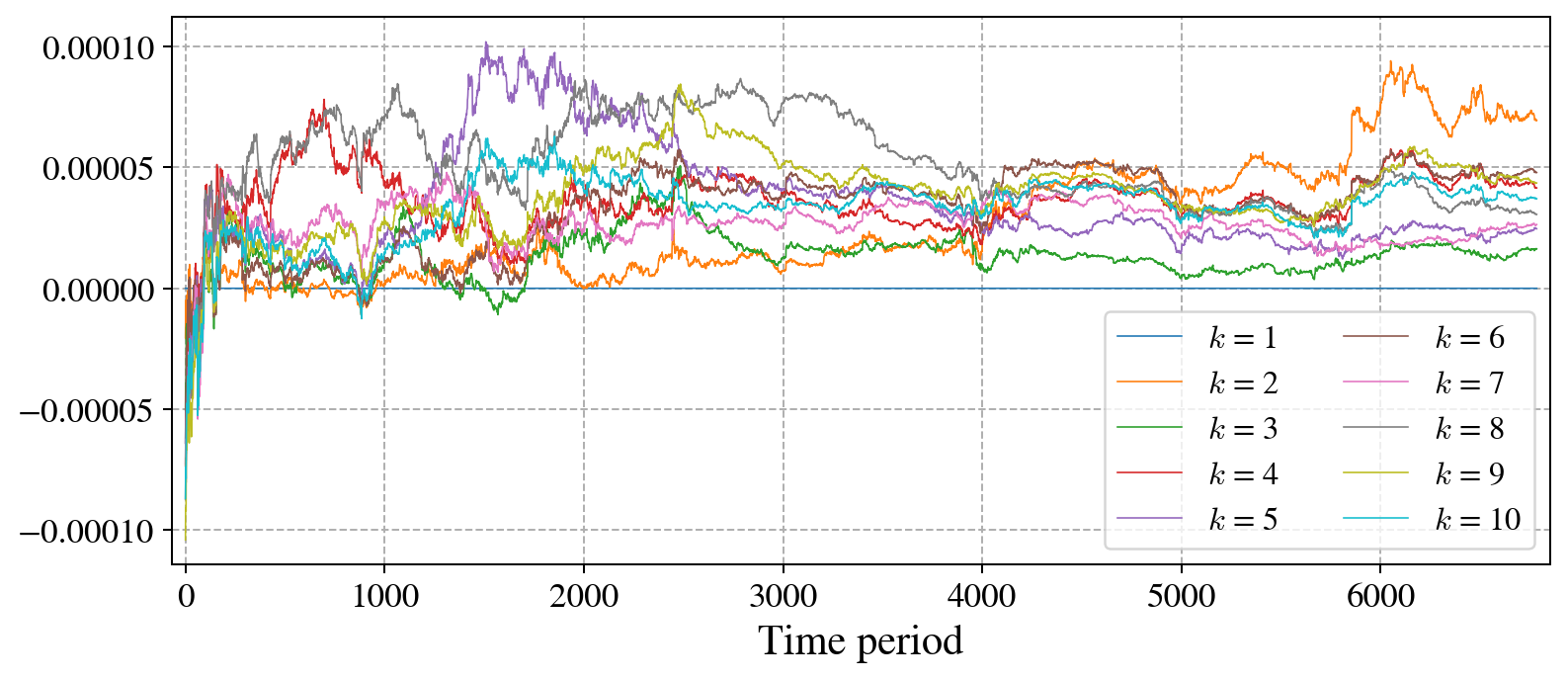}
\par\end{centering}
\caption{Evolutions of the differences $\max_{(b^{1},...,b^{k})\in\mathcal{B}^{m\times k}}\big(W_{n}\big(b_{n}^{k}\big)-W_{n}\big(b_{n}^{k\text{-PUP}}\big)\big)$
(left graphic) and $W_{n}\big(b_{n}^{k\text{-PUP}}\big)-W_{n}\big(b_{n}^{1\text{-PUP}}\big)$
(right graphic) over the investment period.\label{figure 3}}
\end{figure}
\smallskip

\textbf{Discrepancy between the decomposed subsequences of assets'
returns}. Since the $k$ subsequences of assets' returns decomposed
by the algorithm of the $k$-PUP strategy exhibit different statistical
characteristics, the $k$ corresponding best constant strategies defined
on each subsequence yield different averages and variances of returns,
reflecting the optimal utilization of this discrepancy. Indeed, Figure
\ref{figure 4} demonstrates the advantage introduced by sequence
decomposition, which enables substantial improvement compared to using
the original sequence as a whole, i.e., the case $k=1$. The best
$k$-CC strategy, which achieves much higher cumulative wealth when
$k>1$ as shown in Figure \ref{figure 2}, is associated with most
of the corresponding best constant strategies defined on the $k$
subsequences yielding higher averages and variances of returns than
those yielded by the best $1$-CC strategy. In particular, for $k=2$
and $k=6$, all the average returns yielded by the best constant strategies
on the respective subsequences are higher than that for the case $k=1$,
which partly accounts for the exceeding of the $2$-PUP and $6$-PUP
strategies over the best $1$-CC strategy in terms of final cumulative
wealth. Hence, these empirical observations suggest an assessment
that a greater number of decomposed subsequences facilitates a better
improvement in the average returns of the best constant strategies
defined on them.
\begin{figure}[H]
\centering{}\includegraphics[scale=0.36]{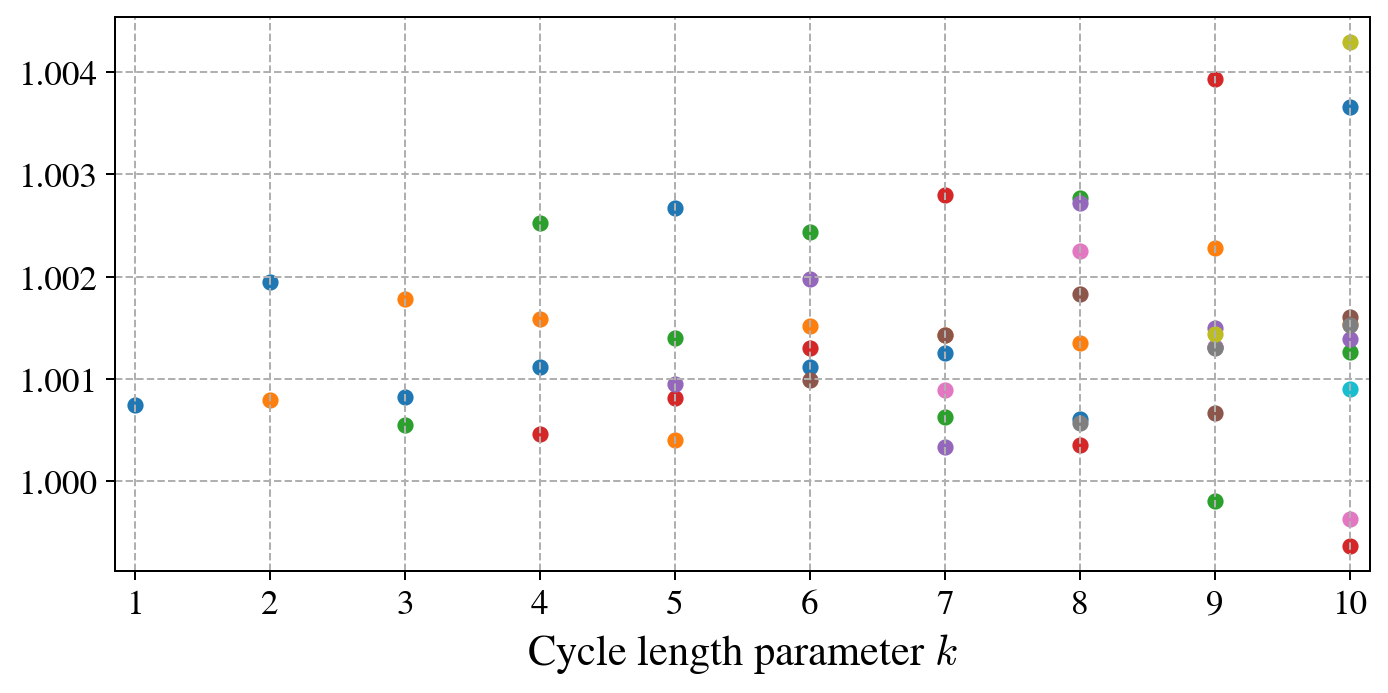}\includegraphics[scale=0.36]{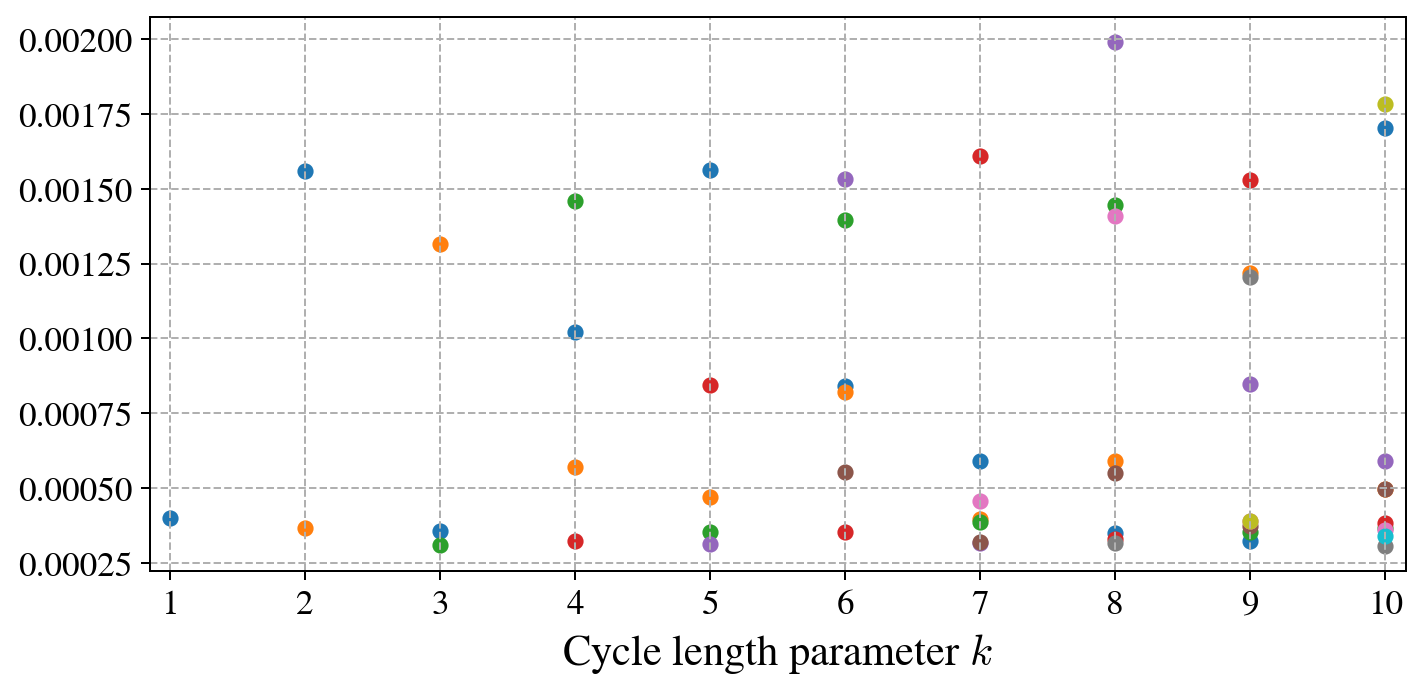}\caption{Average and variance of returns of the best constant strategies defined
on the decomposed subsequences of assets' returns corresponding to
each cycle length parameter $k$. For each $k$, the dots of the same
color in the left and right graphics correspond to the same subsequence.\label{figure 4}}
\end{figure}

\section{Summary and Concluding Remarks}

In this paper, we extended Cover's Universal Portfolio to a parallel
learning algorithm for constructing a dynamic strategy that is consistent
in asymptotic growth rate with the best $k$-cyclic constant strategy,
which is the generalization of the best constant strategy with a longer
repeated cycle, in hindsight as time increases to infinity, regardless
of the realizations of the assets' returns sequence. The proposed
learning algorithm decomposes the original sequence of assets' returns
into multiple subsequences, transforming the original repeated game
into a simultaneously repeated game. This mechanism enables the constructed
dynamic strategy to take advantage solely of the different statistical
properties of the decomposed subsequences, without requiring side
information during the learning process, to eventually exceed the
cumulative wealth of the Universal Portfolio strategy. Furthermore,
while the cumulative wealth yielded by the best constant strategy
inherently exceeds that yielded by the Universal Portfolio strategy
with a large gap over time, it can eventually be significantly outperformed
by that yielded by the strategy constructed according to the proposed
learning algorithm. In addition, when the assets' returns exhibit
serial dependence as a block-wise i.i.d. process, the strategy constructed
according to the generalized Kelly criterion, which is the $k$-cyclic
constant strategy defined by the $k$-log-optimal portfolios, yields
the highest asymptotic growth rate. When the joint distribution of
the block-wise i.i.d. process is unknown, rendering the $k$-log-optimal
portfolios undetermined and making learning algorithms designed for
stationary processes inapplicable, the proposed learning algorithm
generates a strategy that still asymptotically grows to the highest
rate using only the past realizations of the assets' returns.\smallskip

The results of the long-term experiments with real market data for
several values of the parameter $k$ demonstrate the theoretical guarantee
that the proposed learning algorithm can generate a strategy that
eventually yields higher cumulative wealth and growth rate than those
yielded by the Universal portfolio strategy for all cycle lengths
$k$, and also higher than those yielded by the best constant strategy,
provided that $k$ is properly tuned. In general, the cumulative wealth
yielded by the constructed strategy increases as the cycle length
increases in multiples of $k$, yet its empirical performance depends
on the investment horizon, since the convergence speed of its growth
rate to that of the corresponding best $k$-cyclic constant strategy
is also slower. Specifically, the parameter $k$ should be chosen
small enough so that the resulting number of cycles is sufficiently
large relative to the investment horizon, such as $k=2$ or $k=4$
as in the experiments. The outperformance in asymptotic growth rate
of the constructed strategy corresponding to $k>1$ over the one corresponding
to $k=1$, i.e., the Universal Portfolio strategy, invalidates the
classical Kelly criterion in the experimental context and highlights
the benefit of decomposing the assets' returns sequence into multiple
subsequences for improving strategy performance. In particular, as
more subsequences are created when the parameter $k$ increases, the
discrepancy between their characteristics enables the best constant
strategies defined on each of them to achieve significantly higher
average returns compared to that of the best constant strategy defined
on the original sequence.\smallskip

\textbf{Additional notes on practice and other applications}. It is
useful to first recall that the proposed algorithm and the associated
$k$-cyclic constant strategy concept can be applied to any trading
frequency, not limited to the daily frequency used in the experiments.
In practical application, given a sufficiently long investment or
trading horizon, suppose either $k=2$ or $k=3$ is a good choice,
as suggested by the experimental results; then an investor may set
the parameter $k=6$ for the learning algorithm, since the constructed
strategy can guarantee a growth rate similar to that of the strategy
corresponding to the hypothetical $k$ over time. However, if there
are many possible choices for $k$, choosing the cycle length as the
product of all candidate values leads to a very small number of created
cycles relative to the full investment horizon, and the consistency
in growth rate between the constructed strategy and the best $k$-cyclic
constant strategy will manifest significantly more slowly. To address
this issue more efficiently, the investor may construct strategies
with respect to different values of $k$ and combine them into a single
strategy using either the ensembling method proposed in \citet{Lam2024a}
or the exponential weighted average approach as in \citet{CesaBianchi2006},
which guarantees an asymptotic growth rate matching that of the best
strategy corresponding to the optimal choice of $k$, while achieving
a faster convergence speed than using the product of all possible
values. This solution is particularly useful in applying the generalized
Kelly criterion when both the joint distribution and the block length
$k$ of the block-wise i.i.d. process of assets' returns are unknown.
Furthermore, beyond strategy construction, the proposed learning algorithm
and the associated $k$-cyclic constant strategy concept can also
be applied as an ensembling mechanism for combining multiple strategies,
offering a significant improvement over the original ensembling algorithms
and their accelerated variants developed under the distribution-free
preference framework in \citet{Lam2024a}.

\pagebreak

\end{document}